\documentclass[lettersize,journal]{IEEEtran}

\RequirePackage[colorlinks,citecolor=blue,urlcolor=blue,linkcolor=blue]{hyperref}
\usepackage{url}
\usepackage[export]{adjustbox}

\DeclareSymbolFont{slenderlargesymbols}{OMX}{ccex}{m}{n}
\DeclareMathSymbol{\prod}{\mathop}{slenderlargesymbols}{"51}
\usepackage{latexsym,amssymb,amsmath,amsthm,graphics,graphicx,float,psfrag, epsfig, color, authblk, enumerate}

\usepackage{pgfplots}
\usepackage{url,amssymb,amsmath,amsthm,amscd,paralist,bbm,wrapfig}
\usepackage{bm}
\usepackage{algorithm}
\usepackage{algorithmic} 
\usepackage[font=small,skip=4pt]{caption}

\usepackage{mathtools}

\usepackage[T1]{fontenc}
\usepackage[numbers]{natbib}
\tikzstyle{vertex}=[circle,black, fill=black, draw, inner sep=0pt, minimum size=6pt]

\usepackage{tikz}
\usetikzlibrary{decorations.pathreplacing}
\usepackage{xcolor}
\usepackage{mathtools}
\definecolor{cof}{RGB}{219,144,71}
\definecolor{pur}{RGB}{186,146,162}
\definecolor{greeo}{RGB}{91,173,69}
\definecolor{greet}{RGB}{52,111,72}
\pgfplotsset{compat=1.14}

\newtheorem{innercustomgeneric}{\customgenericname}
\providecommand{\customgenericname}{}
\newcommand{\newcustomtheorem}[2]{%
  \newenvironment{#1}[1]
  {%
   \renewcommand\customgenericname{#2}%
   \renewcommand\theinnercustomgeneric{##1}%
   \innercustomgeneric
  }
  {\endinnercustomgeneric}
}

\newcustomtheorem{customasmp}{Assumptions}

\usepackage{latexsym,amssymb,amsmath,amsthm,graphics,graphicx,float,psfrag, epsfig, color, authblk, enumerate}

\usepackage{url,amssymb,amsmath,amsthm,amscd,paralist,bbm,wrapfig}
\usepackage{bm}
\usepackage{algorithm}
\usepackage{algorithmic} 

\usepackage{mathtools}

\usepackage[numbers]{natbib}

\newtheorem{thm}{Theorem}

\newtheorem{defn}[thm]{Definition}

\newtheorem{prop}{Proposition}

\newcommand{\R}{\mathbb{R}}

\newcommand{\E}{\mathbb{E}}

\newcommand{\Lcal}{\mathcal{L}}
\newcommand{\Ncal}{\mathcal{N}}
\DeclareMathOperator*{\argmax}{arg\,max}
\DeclareMathOperator*{\argmin}{arg\,min}
\DeclareMathOperator*{\KL}{KL}
\DeclareMathOperator*{\ELBO}{ELBO}
\DeclareMathOperator*{\ELBOProxy}{ELBOProxy}

\DeclarePairedDelimiterX{\infdivx}[2]{(}{)}{%
  #1\;\delimsize\|\;#2%
}
\newcommand{\infdiv}{D_{\KL}\infdivx}

\newtheorem{lemma}{Lemma}

\usepackage[resetlabels]{multibib}
\newcites{Phys}{References}
\begin{document}

\title{Discovering Structure From Corruption for Unsupervised Image Reconstruction} %

\author{Oscar Leong$^*$, Angela F. Gao$^*$, He Sun, and Katherine L. Bouman
\thanks{The authors are with the Computing and Mathematical Sciences Department at the California Institute of Technology (Caltech), Pasadena, California. Katherine L. Bouman is also with the Electrical Engineering, Astronomy, and Mechanical Engineering Departments at Caltech. He Sun is also with the College of Future Technology and National Biomedical Imaging Center at Peking University, Beijing, China. Most of this work was done while He Sun was at Caltech. This research was carried out at the Jet Propulsion Laboratory and Caltech under a contract with the National Aeronautics and Space Administration and funded through the PDRDF. In addition, 
this work was sponsored by NSF Awards 2048237, 1935980, and an Amazon AI4Science Partnership Discovery Grant. We would like to thank Ben Prather, Abhishek Joshi, Vedant Dhruv, Chi-kwan Chan, and Charles Gammie for providing black hole simulations used in this work. We would also like to thank Aviad Levis, Yu Sun, and Jorio Cocola for their feedback and guidance.}}


\setlength{\abovedisplayskip}{6pt}
\setlength{\belowdisplayskip}{6pt}

\maketitle

\def\thefootnote{*}\footnotetext{These authors contributed equally. Corresponding emails: \href{mailto:oleong@caltech.edu}{oleong@caltech.edu}, \href{mailto:afgao@caltech.edu}{afgao@caltech.edu}.}

\begin{figure*}
    \centering
    \includegraphics[width=\textwidth]{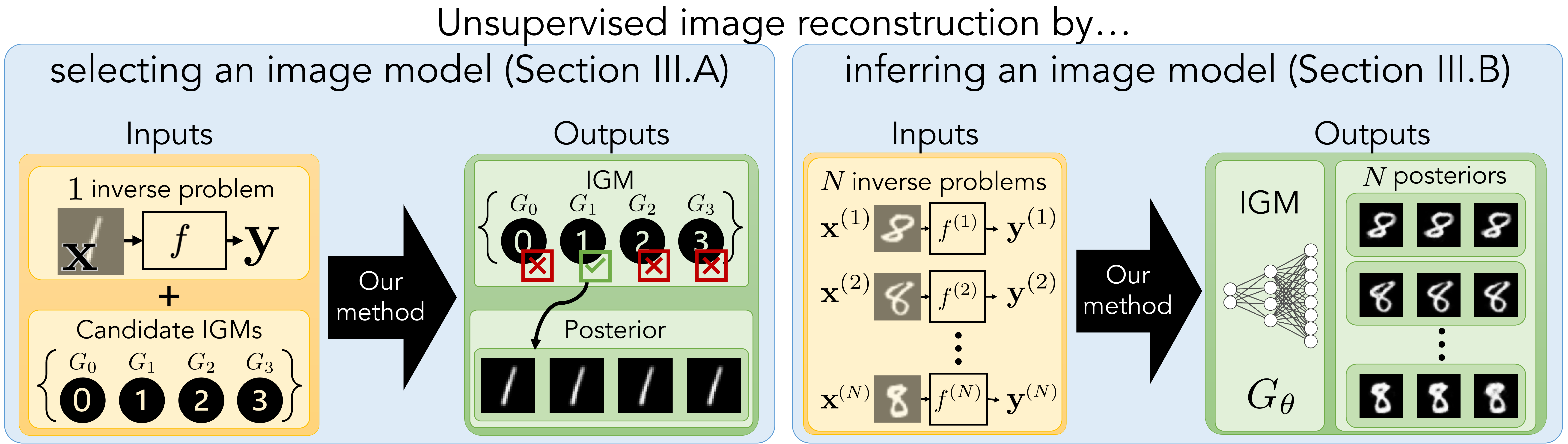}
    \caption{\textbf{Method overview.} In this paper we tackle ill-posed image reconstruction problems when a traditional image prior is not available or cannot be derived from example images. The key idea of our work is that image reconstruction is possible when one has access to only corrupted measurement examples and the underlying images share common, low-dimensional structure. We explore such ideas in this work in the following two ways: \textbf{Left:} Given a single measurement example (with the corresponding forward model) and a set of candidate image generation models (IGMs), we use our proposed criterion to select the best IGM from the measurements alone and then recover an image posterior under that IGM. The criterion we use for model selection is an approximation of the evidence lower bound (i.e., ELBOProxy in Eq. \eqref{eqref:ELBOProxy})  (see Sec.~\ref{sec:ELBO-intro}). \textbf{Right:} Building off of this, we consider the setting where we must infer the IGM from noisy measurements alone (see Sec.~{\ref{sec:learning}}). In this case, we solve $N$ inverse problems simultaneously. The inputs of this method are $N$ measurement examples with their known forward models, and the outputs are a single inferred IGM and $N$ latent embeddings that, when combined, lead to $N$ image reconstruction posteriors.}
    \label{fig:method_both}
\end{figure*}

\begin{abstract}

We consider solving ill-posed imaging inverse problems without access to an image prior or ground-truth examples.
An overarching challenge in these inverse problems is that an infinite number of images, including many that are implausible, are consistent with the observed measurements.
Thus, image priors are required to reduce the space of possible solutions to more desirable reconstructions. However, in many applications it is difficult or potentially impossible to obtain example images to construct an image prior. Hence inaccurate priors are often used, which inevitably result in biased solutions.
Rather than solving an inverse problem using priors that encode the spatial structure of any one image, we propose to solve a set of inverse problems jointly by incorporating prior constraints on the collective structure of the underlying images. 
The key assumption of our work is that the underlying images we aim to reconstruct share common, low-dimensional structure. We show that such a set of inverse problems can be solved simultaneously without the use of a spatial image prior by instead inferring a shared image generator with a low-dimensional latent space.
The parameters of the generator and latent embeddings are found by maximizing a proxy for the Evidence Lower Bound (ELBO). Once identified, the generator and latent embeddings can be combined to provide reconstructed images for each inverse problem.
The framework we propose can handle general forward model corruptions, and we show that measurements derived from only a small number of ground-truth images ($\leqslant 150$) are sufficient for image reconstruction. We demonstrate our approach on a variety of convex and non-convex inverse problems, including  denoising, phase retrieval, and black hole video reconstruction.

\end{abstract}

\vspace{-5mm}

\section{Introduction}

In imaging inverse problems, the goal is to recover an underlying image from corrupted measurements, where the measurements and image are related via an understood forward model: $y = f(x) + \eta$. Here, $y$ are measurements, $x$ is the underlying image, $f$ is a known forward model, and $\eta$ is noise. Such problems are ubiquitous and include denoising \cite{Burger2012, dyer2022lightsheet}, super-resolution \cite{CandesGranda}, compressed sensing \cite{CandesRombergTao2006, Donoho2006}, phase retrieval \cite{Numerics-of-PR}, and deconvolution \cite{Lietal2016}. 
Due to corruption by the forward model and noise, these problems are often ill-posed: there are many images that are consistent with the observed measurements, including ones that are implausible.

To combat the ill-posedness in imaging  problems, solving for an image traditionally requires imposing additional structural assumptions to reduce the space of possible solutions. 
We encode these assumptions in an \textit{image generation model (IGM)}, whose goal is to capture the desired properties of an image's spatial structure. 
IGMs are general; they encompass probabilistic \textit{spatial-domain} priors (e.g., that encourage smoothness or sparsity), but also include deep image generators that are not necessarily probabilistic but are trained to primarily sample a certain class of images.

In order to define an IGM, it is necessary to have knowledge of the underlying image's structure. 
If images similar to the underlying image are available, then an
IGM can be learned directly \cite{romano2017RED, venkatakrishnan2013pnp, Boraetal17}. However, an abundance of clean images is not available for many scientific imaging modalities (e.g., geophysical imaging and astronomical imaging). Collecting images in these domains can be extremely invasive, time-consuming, expensive, or even impossible. For instance, how should we define an IGM for black hole imaging without having ever seen a direct image of a black hole or knowing what one should look like? Moreover, classical approaches that utilize hand-crafted IGMs, such as total variation \cite{TV-ROF} or sparsity in a wavelet basis \cite{mallat1999wavelet}, are prone to human bias \cite{Levinetal09}.

In this work, we show how one can solve a set of ill-posed image reconstruction tasks in an unsupervised fashion, i.e., without prior information about an image's spatial structure or access to clean, example images. The key insight of our work is that knowledge of common structure across multiple diverse images can be sufficient regularization alone. In particular, suppose we have access to a collection of noisy measurements $\{y^{(i)}\}_{i=1}^N$ that are observed through (potentially different) forward models $y^{(i)} := f^{(i)}(x^{(i)}) + \eta^{(i)}$. 
The core assumption we make is that the different underlying images $\{x^{(i)}\}_{i=1}^N$ are drawn from the same distribution (unknown {\it a priori}) and share common, low-dimensional structure. Thus, our ``prior'' is not at the spatial-level, but rather exploits the \textit{collective structure} of the underlying images. This assumption is satisfied in a number of applications where there is no access to an abundance of clean images. For instance, although we might not know what a black hole looks like, we might expect it to be similar in appearance over time. We show that under this assumption, the image reconstruction posteriors $p(x|y^{(i)})$ can be learned jointly from a small number of examples $\{y^{(i)}\}_{i=1}^N$ due to the common, low-dimensional structure of the collection $\{x^{(i)}\}_{i=1}^N$. Specifically, our main result is that one can capitalize on this common structure by jointly solving for 1) a shared image generator $G_{\theta}$ and 2) $N$ low-dimensional latent distributions $q_{\phi^{(i)}}$, such that the distribution induced by the push-forward of $q_{\phi^{(i)}}$ through $G_{\theta}$ approximately captures the image reconstruction posterior $p(x| y^{(i)})$ for each measurement example $i \in [N]$.

\vspace{-2mm}

\subsection{Our Contributions}

We outline the main contributions of our work, which extends our prior work presented in \cite{gao2023image}:
\begin{enumerate}
    \item We solve a collection of ill-posed inverse problems without prior knowledge of an image's spatial structure by exploiting the common, low-dimensional structure shared across images. This common structure is exploited when inferring a shared IGM with a low-dimensional latent space.
    \item To infer this IGM, we define a loss inspired by the evidence lower bound (ELBO).
    We motivate this loss by showing how it aids in unsupervised image reconstruction by helping select one IGM from a collection of candidate IGMs using a single measurement example.
    
    
    \item 
    We apply our approach to convex and non-convex inverse problems, such as denoising, black hole compressed sensing, and phase retrieval. We establish that we can solve inverse problems without spatial-level priors and demonstrate good performance with only a small number of independent measurement examples (e.g., $\leqslant 150$).
\item We theoretically analyze the inferred IGM in linear inverse problems under a linear image model to show that in this setting the inferred IGM performs dimensionality reduction akin to PCA on the collection of measurements. 
\end{enumerate}


\vspace{-3mm}
\section{Background and Related Work}

We now discuss related literature in model selection and learning-based IGMs. In order to highlight our key contributions, we emphasize the following assumptions in our framework:

\begin{enumerate}
    \item We do not have access to a set of images from the same distribution as the underlying images.
    \item We only have access to a collection of measurement examples, where each example comes from a different underlying image. The number of examples $N$ is small, e.g., $N \leqslant 150$.
    \item For each underlying image $x^{(i)}$ we wish to reconstruct, we only have access to a single measurement example $y^{(i)} = f^{(i)}(x^{(i)}) + \eta^{(i)}$. That is, we do not have multiple observations of the same underlying image. Note each $f^{(i)}$ can be potentially different.
\end{enumerate} 

\vspace{-4mm}
\subsection{Model selection}
Model selection techniques seek to choose a model that best explains data by balancing performance and model complexity. In supervised learning problems with sufficiently large amounts of data, this can be achieved simply by evaluating the performance of different candidate models using reserved test data \cite{stone1974cross}. However, in image reconstruction or other inverse problems with limited data, one cannot afford to hold out data. In these cases, model selection is commonly conducted using probabilistic metrics. The simplest probabilistic metric used for linear model selection is adjusted R$^2$ \cite{Miles-Rsquared}. It re-weights the goodness-of-fit by the number of linear model parameters, helping reject high-dimensional parameters that do not improve the data fitting accuracy. Similar metrics in nonlinear model selection are Bayesian Information Criterion (BIC) \cite{BIC} and Akaike Information Criterion (AIC) \cite{akaike1974new}. AIC and BIC compute different weighted summations of a model's log-likelihood and complexity, offering different trade-offs between bias and variance to identify the best model for a given dataset. 

In our work, we consider the use of the ELBO as a model selection criterion. In \cite{Abdellatif18, AbdellatifAlquier18}, the use of the ELBO as a model selection criterion is theoretically analyzed and rates of convergence for variational posterior estimation are shown. Additionally, \cite{Taoetal18} proposes a generalized class of evidence lower bounds leveraging an extension of the evidence score. 
In \cite{Sunetal-alphaDPI}, the ELBO is used for model selection to select a few, discrete parameters modeling a physical system (e.g., parameters that govern the orbit of an exoplanet). A significant difference in our context, however, is that we use the ELBO as a model selection criterion in a \textit{high-dimensional} imaging context, and we optimize the ELBO over a continuous space of possible parameters.

\vspace{-4mm}
\subsection{Learning IGMs}
With access to a large corpus of example images, it is possible to directly learn an IGM to help solve inverse problems. Seminal work along these lines utilizing generative networks showcased that a pre-trained Generative Adversarial Network (GAN) can be used as an IGM in the problem of compressed sensing \cite{Boraetal17}. To solve the inverse problem, the GAN was used to constrain the search space for inversion.  This approach was shown to outperform sparsity-based techniques with 5-10x fewer measurements. Since then, this idea has been expanded to other inverse problems, including denoising \cite{Heckeletal2018}, super-resolution \cite{PULSE_CVPR_2020}, magnetic resonance imaging (MRI) \cite{Mardani2017, Songetal22}, and phase retrieval \cite{HLV18, ShamshadAhmed21}. However, the biggest downside to this approach is the requirement of a large dataset of example images similar to the underlying image, which is often difficult or impossible to obtain in practice. Hence, we consider approaches that are able to directly solve inverse problems without example images.

Methods that aim to learn an IGM from only noisy measurements have been proposed. The main four distinctions between our work and these methods are that these works either: 1) require multiple independent observations of the \textit{same} underlying image, 2) can only be applied to certain inverse problems, 3) require significantly more observations (either through more observations of each underlying image or by observing more underlying images), or 4) require significant hyperparameter tuning based on knowledge of example images.


\noindent \textbf{Noise2Noise (N2N)} \cite{Lehtinenetal18} learns to denoise images by training on a collection of noisy, independent observations of the \textit{same} image. To do so, N2N learns a neural network $\Phi_{\theta}$ whose goal is to map between noisy images $y$ and denoised images $x$. Since it has no denoised image examples to supervise training, it instead employs a loss that maps between noisy examples of the same underlying image. This objective is as follows:
\begin{equation}
\argmin_{\theta}
\sum_{i=1}^N \E_{y^{(i)}_1 \sim \mathcal{Y}_1} \E_{y^{(i)}_2 \sim \mathcal{Y}_2}[L(\Phi_{\theta}(y^{(i)}_1), y^{(i)}_2)],
\label{eq:n2n}
\end{equation}
where $y^{(i)}_j$ corresponds to a noisy observation of the $i$-th underlying image $x^{(i)}$, and $\mathcal{Y}_j$ is a distribution of noisy images where $\E_{y \sim \mathcal{Y}_j}[y] = x$. 
This N2N objective requires at least two observations of the same image and is limited by the assumption that the expected value of multiple observations of a single image is the underlying image.
Thus, N2N is only applicable to denoising problems where the forward model is the identity matrix with independent noise on each pixel.
Additionally, in practice N2N requires thousands of underlying images (i.e., $N = O(1000)$) to perform well. 
Thus, N2N's main distinctions with our work are distinctions 1), 2), and 3).


\noindent \textbf{Regularization by Artifact Removal (RARE)} \cite{liu2020rare} generalizes N2N to perform image reconstruction from measurements under linear forward models. That is, the objective in Eq.~\eqref{eq:n2n} is modified to include a pseudo-inverse. Nonetheless, multiple observations of the same underlying image are required, such that $\E_{y \sim \mathcal{Y}}[A^{\dagger}y] = x$ for the pseudo-inverse matrix $A^{\dagger}$. Thus, RARE suffers from the same limiting distinctions as N2N (i.e., 1), 2), and 3)).

\noindent \textbf{Noise2Void} \cite{Krulletal19} \textbf{and Noise2Self} \cite{batson2019noise2self} assume that the image can be partitioned such that the measurement noise in one subset of the partition is independent conditioned on the measurements in the other subset. This is true for denoising, but not applicable to general forward models. For example, in black hole and MRI compressed sensing, it is not true that the measurement noise can be independently partitioned since each measurement is a linear combination of all pixels. While this makes Noise2Void and Noise2Self more restrictive in the corruptions they can handle compared to RARE, they also don't require multiple observations of the same underlying image. Hence the main differences between these works and our own are distinctions 2) and 3). 


\noindent \textbf{AmbientGAN} \cite{bora2018ambientgan} and other similar approaches based on GANs \cite{kabkab2018task} and Variational Autoencoders (VAEs) \cite{Olsenetal22, Mendozaetal22} have been proposed to learn an IGM directly from noisy measurements. For instance, AmbientGAN aims to learn a generator whose images lead to simulated measurements that are indistinguishable from the observed measurements; this generator can subsequently be used as a prior to solve inverse problems.
However, AmbientGAN requires many measurement examples (on the order of 10,000) to produce a high quality generator.
We corroborate this with experiments in Section \ref{sec:results-joint} to show that they require many independent observations and/or fine tuning of learning parameters to achieve good performance. Thus, the main distinctions between AmbientGAN and our work are 3) and 4).


\noindent \textbf{Deep Image Prior (DIP)} \cite{ulyanov2018deepjournal} uses a convolutional neural network as an implicit ``prior''. DIP has shown strong performance across a variety of inverse problems to perform image reconstruction without explicit probabilistic priors. However, it is prone to overfitting and requires selecting a specific stopping criterion. While this works well when example images exist, selecting this stopping condition from noisy measurements alone introduces significant human bias that can negatively impact results. Thus, the main distinction between DIP and our work is 4).

We would also like to highlight additional work done to improve certain aspects of the DIP. While the original DIP method used a U-Net architecture \cite{U-Net-paper}, other works such as the Deep Decoder \cite{HH2018} and ConvDecoder \cite{DarestaniHeckel} used a decoder-like architecture that progressively grows a low-dimensional random tensor to a high-dimensional image. When underparameterized, such architectures have been shown to avoid overfitting and the need for early stopping. Other works mitigating early stopping include \cite{Chengetal19}, which takes a Bayesian perspective to the DIP by using Langevin dynamics to perform posterior inference over the weights to improve performance and show this does not lead to overfitting.


\vspace{-4mm}
\section{Approach}
\begin{figure*}[ht]
    \centering
    \includegraphics[width=0.95\textwidth]{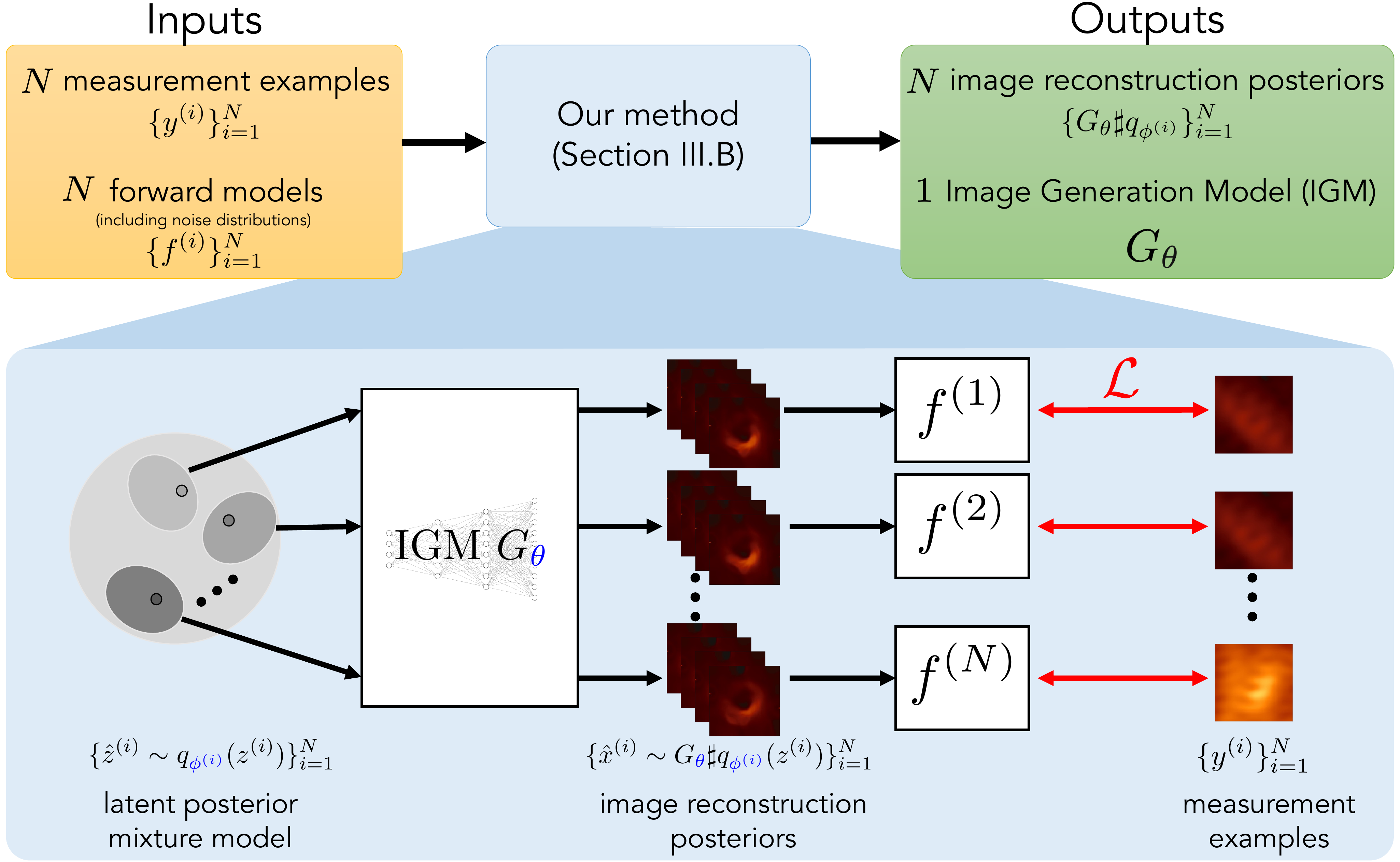}
     \caption{We propose to solve a collection of $N$ ill-posed inverse problems by exploiting the common, low-dimensional structure of the underlying images. Given a set of $N$ measurement examples $\{y^{(i)}\}_{i=1}^N$ from $N$ different underlying images, we propose to model each image posterior as the output of a shared IGM with a low-dimensional latent space. In particular, each posterior is approximated by $G_{\theta}\sharp q_{\phi^{(i)}}$, the push-forward of $q_{\phi^{(i)}}$ through $G_{\theta}$, where $G_{\theta}$ is the shared, common generator to all $N$ examples and $q_{\phi^{(i)}}$ is a low-dimensional, variational distribution particular to the $i$-th example. The parameters, $\theta$ and $\{\phi^{(i)}\}_{i=1}^N$, are colored in blue and are jointly inferred. That is, such parameters are inferred specifically for the measurements $\{y^{(i)}\}_{i=1}^N$. The loss we use  is the negative $\ELBOProxy$, which is denoted by $\mathcal{L}$ and given by Eq. \eqref{eqref:learning-objective}. Note that there is no notion of a training set and test set as we aim to solve the inverse problems jointly from all available measurements.
     }
    \label{fig:diagram}
\end{figure*}
In this work, we propose to solve a set of inverse problems without prior access to an IGM by assuming that the set of underlying images have common, low-dimensional structure. We motivate the use of optimizing the ELBO to infer an IGM by showing that it is a good criterion for generative model {\it selection} in Section~\ref{sec:ELBO-intro}.
Then, by optimizing the ELBO, we show in Section~\ref{sec:learning} that one can directly {\it infer}  an IGM from corrupted measurements alone by parameterizing the image model as a deep generative network with a low-dimensional latent distribution. 
The IGM network weights are shared across all images, capitalizing on the common structure present in the data, while the parameters of each latent distribution are learned jointly with the generator to model the image posteriors for each measurement example. 
\subsection{Motivation for ELBO as a model selection criterion} \label{sec:ELBO-intro}

In order to accurately {\it infer} an IGM, we motivate the use of the ELBO as a loss by showing that it provides a principled criterion for {\it selecting} an IGM to use as a prior model. Suppose we are given noisy measurements from a single image: $y = f(x) + \eta$. In order to reconstruct the image $x$, we traditionally first require an IGM $G$ that captures the distribution $x$ was sampled from. A natural approach would be to find or select the model $G$ that maximizes the model posterior distribution $p(G |y) \propto p(y | G)p(G).$
That is, conditioned on the noisy measurements, find the IGM of highest likelihood. Unfortunately computing $p(y | G)$ is intractable, as it requires marginalizing and integrating over all $x$ encompassed by the IGM $G$. However, we show that this quantity can be well approximated using the ELBO.


To motivate our discussion, we first consider estimating the image posterior $p(x|y,G)$ by learning the parameters $\phi$ of a variational distribution $h_{\phi}(x)$. Observe that the definition of the KL-divergence followed by an application of Bayes' theorem gives \begin{align*}
    &\infdiv{h_{\phi}(x)}{p(x|y,G)}  := \mathbb{E}_{x \sim h_{\phi}(x)}\left[\log \frac{h_{\phi}(x)}{p(x|y,G)}
     \right] \\
    & = \mathbb{E}_{x \sim h_{\phi}(x)}\left[\log \frac{h_{\phi}(x)p(y|G)}{p(y|x,G)p(x|G)}
     \right]\\
    & = - \mathbb{E}_{x \sim h_{\phi}(x)}[\log p(y|x,G)  + \log p(x|G) - \log h_{\phi}(x)]\\
    & \qquad \qquad +\log p(y|G).
\end{align*} The ELBO of an IGM $G$ given measurements $y$ under variational distribution $h_{\phi}$ is defined by \begin{align}
    \ELBO(G, h_{\phi};y) & := \mathbb{E}_{x \sim h_{\phi}(x)}[\log p(y|x,G) \nonumber\\
    & + \log p(x|G) - \log h_{\phi}(x)].
    \label{eq:ELBO} 
\end{align} Rearranging the previous equation, we see that by the non-negativity of the KL-divergence that \begin{align*}
    \log p(y|G) & = \infdiv{h_{\phi}(x)}{p(x|y,G)}  \\
    & + \ELBO(G, h_{\phi};y) \\
    & \geqslant \ELBO(G, h_{\phi};y).
\end{align*}
Thus, we can lower bound the model posterior as \begin{align*}
    \log p(G|y) \geqslant \ELBO(G, h_{\phi};y) + \log p(G) - \log p(y).
\end{align*} 
Note that $\log p(y)$ is independent of the parameters of interest, $\phi$. If the variational distribution $h_{\phi}(x)$ is a good approximation to the posterior $p(x|y,G)$, $D_{\mathrm{KL}} \approx 0$. Thus, maximizing $\log p(G|y)$ with respect to $G$ is approximately equivalent to maximizing $\ELBO(G, h_{\phi};y) + \log p(G)$.


 Each term in the ELBO objective encourages certain properties of the IGM $G$. In particular, the first term in the ELBO, $\E_{x \sim h_{\phi}(x)}[\log p(y|x,G)]$, requires that $G$ should lead to an image estimate that is consistent with our measurements $y$. 
 The second term, $\E_{x \sim h_{\phi}(x)}[\log p(x|G)]$, encourages images sampled from $h_{\phi}(x)$ to have high likelihood under our model $G$. 
 The final term is the entropy term, $\E_{x \sim h_{\phi}(x)} [-\log h_{\phi}(x)]$, which encourages a $G$ that leads to ``fatter'' minima that are less sensitive to small changes in likely images $x$ under $G$.

\subsubsection{ELBOProxy} 

Some IGMs are explicit, which allows for direct computation of $\log p(x|G)$. For example, if our IGM models $x$ as isotropic Gaussian with variance $\lambda$, then $-\log p(x|G) \propto \lambda^{-1}\|x\|^2_2$. In this case, we can optimize the ELBO defined in Equation~\eqref{eq:ELBO} directly and then perform model selection. However, an important class of IGMs that we are interested in are those given by deep generative networks. Such IGMs are not probabilistic in the usual Bayesian interpretation of a prior, but instead implicitly enforce structure in the data. A key characteristic of many generative network architectures (e.g., VAEs and GANs) that we leverage is that they generate high-dimensional images from low-dimensional latent representations. Bottlenecking helps the network learn global characteristics of the underlying image distribution while also respecting the low intrinsic dimensionality of natural images. 
However, this means that we can only compute $\log p(x|G)$ directly if we have an injective map \cite{kothari2021trumpets}. This architectural requirement limits the expressivity of the network.

We instead consider a proxy of the ELBO that is especially helpful for deep generative networks. That is, suppose our IGM is of the form $x = G(z)$. Introducing a variational family for our latent representations $z \sim q_{\phi}(z)$ and choosing a latent prior distribution $\log p_Z(z| G)$, we arrive at the following proxy of the ELBO:
\begin{align}
\ELBOProxy(G, q_{\phi}&;y)  := \E_{z \sim q_{\phi}(z)}[\log p(y |G(z)) \nonumber\\
& + \log p_Z(z | G) - \log q_{\phi}(z)]. \label{eqref:ELBOProxy}  
\end{align} In our experiments, we chose $p_Z(z|G)$ to be an isotropic Gaussian prior. This is a common choice in many generative modeling frameworks and has shown to be a good choice of prior in the latent space. 

To motivate this proxy, it is instructive to consider the case where our variational distribution $h_{\phi}:= G\sharp q_{\phi}$ is the push-forward of a latent distribution $q_{\phi}$ through an injective or invertible function $G$. To be precise, recall the following definition of the push-forward measure.

\begin{defn}
    Let $G :\R^k \rightarrow \R^n$ be a measurable function and suppose $p$ is a distribution (or, more generally, a measure) on $\R^k$. Then the push-forward measure $\mu := G \sharp p$ is the measure on $\R^n$ that satisfies the following: for all Borel sets $A$ of $\R^n$, $\mu(A) = p(G^{-1}(A))$ where $G^{-1}(A)$ denotes the preimage of $A$ with respect to $G$.
\end{defn}

The push-forward measure essentially characterizes how a distribution $p$ changes when passed through a function $G$. It follows from the definition of the push-forward that $x \sim G\sharp q_{\phi}$ if and only if $x = G(z)$ where $z \sim q_{\phi}$.  In the case $h_{\phi} = G \sharp q_{\phi}$ for an injective function $G$, the $\ELBO$ and $\ELBOProxy$ are equivalent, as shown in the following proposition:


\begin{prop} \label{prop:ELBOProxy_equivalence}
Suppose $G : \R^k \rightarrow \R^n$ is continuously differentiable and injective. For two probability distributions $p_Z$ and $q_{\phi}$ on $\R^k$, define the measures $p(\cdot|G) = G\sharp p_Z$ and $h_{\phi} = G\sharp q_{\phi}$. Then \begin{align*}
    \ELBO(G, h_{\phi};y)  = \ELBOProxy(G,q_{\phi};y)\ \forall y \in \R^m.
\end{align*}
\end{prop}

\begin{proof} It suffices to show \begin{align*}
    \mathbb{E}_{x \sim h_{\phi}(x)}[&\log p(x|G)  - \log h_{\phi}(x)] \\
    & = \mathbb{E}_{z \sim q_{\phi}(z)}[\log p_Z(z|G) - \log q_{\phi}(z)]
\end{align*} Let $J_{G}(z) \in \R^{n \times k}$ denote the Jacobian of $G$ at an input $z \in \R^k$. Since $G$ is injective and continuously differentiable with $p(\cdot|G) = G\sharp p_Z$, we can compute the likelihood of any point $x \in \mathrm{range}(G)$ \cite{kothari2021trumpets} via \begin{align*}
    \log p(x|G) & = \log p_Z(G^{\dagger}(x)|G) \\
   &  - \frac{1}{2}\log|\det[J_G(G^{\dagger}(x))^TJ_G(G^{\dagger}(x))]|
\end{align*} where $G^{\dagger}$ is the inverse of $G$ on its range. This is essentially the classical change-of-variables formula specialized to the case when $G$ is injective and we wish to access likelihoods on the range of $G$. Note that this equation is only valid for points in the range of the injective function $G$. Likewise, since $h_{\phi} = G \sharp q_{\phi}$, we can compute the entropy of $h_{\phi}$ for any point $x \in \mathrm{range}(G)$ as \begin{align*}
    \log h_{\phi}(x) & = \log q_{\phi}(G^{\dagger}(x)) \\
   &  - \frac{1}{2}\log|\det[J_G(G^{\dagger}(x))^TJ_G(G^{\dagger}(x))]|.
\end{align*} 

Now observe that for $x \sim h_{\phi}$, $x \in \mathrm{range}(G)$ as $h_{\phi}$ is the push-forward of $q_{\phi}$ through $G$. Thus, for $x \sim h_{\phi}$, we have that \begin{align*}
    \log p(x|G) & - \log h_{\phi}(x) \\
    & = \log p_Z(G^{\dagger}(x)|G) - \log q_{\phi}(G^{\dagger}(x)).
\end{align*} By the definition of the push-forward measure, we have that $x \sim h_{\phi}$ implies $x = G(z)$ for some $z \sim q_{\phi}$. Using our previous formulas, we can compute the expectation over the difference $\log p(x|G)  - \log h_{\phi}(x)$ with respect to $h_{\phi}$ as \begin{align*}
    & \mathbb{E}_{x \sim h_{\phi}(x)} [\log p(x|G)  - \log h_{\phi}(x)] \\
    & = \mathbb{E}_{x \sim h_{\phi}(x)}[\log p_Z(G^{\dagger}(x)|G) - \log q_{\phi}(G^{\dagger}(x))] \\
    & =\mathbb{E}_{z \sim q_{\phi}(z)}[\log p_Z(G^{\dagger}(G(z))|G) - \log q_{\phi}(G^{\dagger}(G(z)))] \\
    & = \mathbb{E}_{z \sim q_{\phi}(z)}[\log p_Z(z|G) - \log q_{\phi}(z)].
\end{align*}
\vspace{-2mm}
\end{proof}
An important consequence of this result is that for injective generators $G$, the inverse of $G$ (on its range) is not required for computing the $\ELBO$. In this case, the $\ELBOProxy$ is in fact equivalent to the $\ELBO$. While not all generators $G$ will be injective, quality generators are largely injective over high likelihood image samples. 
In Section \ref{sec:mnist-model-selection-exp} and Fig. \ref{fig:MNIST_ELBO_exp}, we experimentally show that this proxy can aid in selecting potentially non-injective generative networks from corrupted measurements.

\subsubsection{Toy example} \label{sec:mnist-model-selection-exp}  To illustrate the use of the $\ELBOProxy$ as a model selection criterion, we conduct the following experiment that asks whether the $\ELBOProxy$ can identify the best model from a given set of image generation models. For this experiment, we use the MNIST dataset \cite{MNIST} and consider two inverse problems: denoising and phase retrieval. We train a generative model $G_c$ on each class $c \in \{0,1,2,\dots,9\}$ using the clean MNIST images directly. Hence, $G_c$ generates images from class $c$ via $G_c(z)$ where $z \sim \mathcal{N}(0,I)$. Then, given noisy measurements $y_{c}$ from a single image from class $c$, we ask whether the generative model $G_c$ from the appropriate class would achieve the best $\ELBOProxy$. Each $G_c$ is the decoder of a VAE with a low-dimensional latent space, with no architectural constraints to ensure injectivity. For denoising, our measurements are $y_c = x_c + \eta_c$ where $ \eta_c\sim \mathcal{N}(0, \sigma^2 I)$ and $\sigma = \sqrt{0.5}$. For phase retrieval, $y_c = |\mathcal{F}(x_{c})| + \eta_c$ where $\mathcal{F}$ is the Fourier transform and $\eta_{c} \sim \mathcal{N}(0, \sigma^2 I)$ with $\sigma = \sqrt{0.05}$.

We construct $10 \times 10$ arrays for each problem, where in the $i$-th row and $j$-th column, we compute the negative $\ELBOProxy$ obtained by using model $G_{{i-1}}$ to reconstruct images from class $j-1$. We calculate $\ELBOProxy(G_{c},q_{\phi_c};y_c)$ by parameterizing $q_{\phi_c}$ with a Normalizing Flow \cite{dinh2016density} and optimizing network weights $\phi_c$ to maximize \eqref{eqref:ELBOProxy}. The expectation in the $\ELBOProxy$ is approximated via Monte Carlo sampling. Results from the first $5$ classes are shown in Fig.~\ref{fig:MNIST_ELBO_exp} and the full arrays are shown in the supplemental materials. We note that all of the correct models are chosen in both denoising and phase retrieval. We also note some interesting cases where the $\ELBOProxy$ values are similar for certain cases, such as when recovering the $3$ or $4$ image. For example, when denoising the $4$ image, both $G_{4}$ and $G_{9}$ achieve comparable $\ELBOProxy$ values. By carefully inspecting the noisy image of $4$, one can see that both models are reasonable given the structure of the noise. 

\begin{figure}
    \centering
    \includegraphics[width=0.48\textwidth]{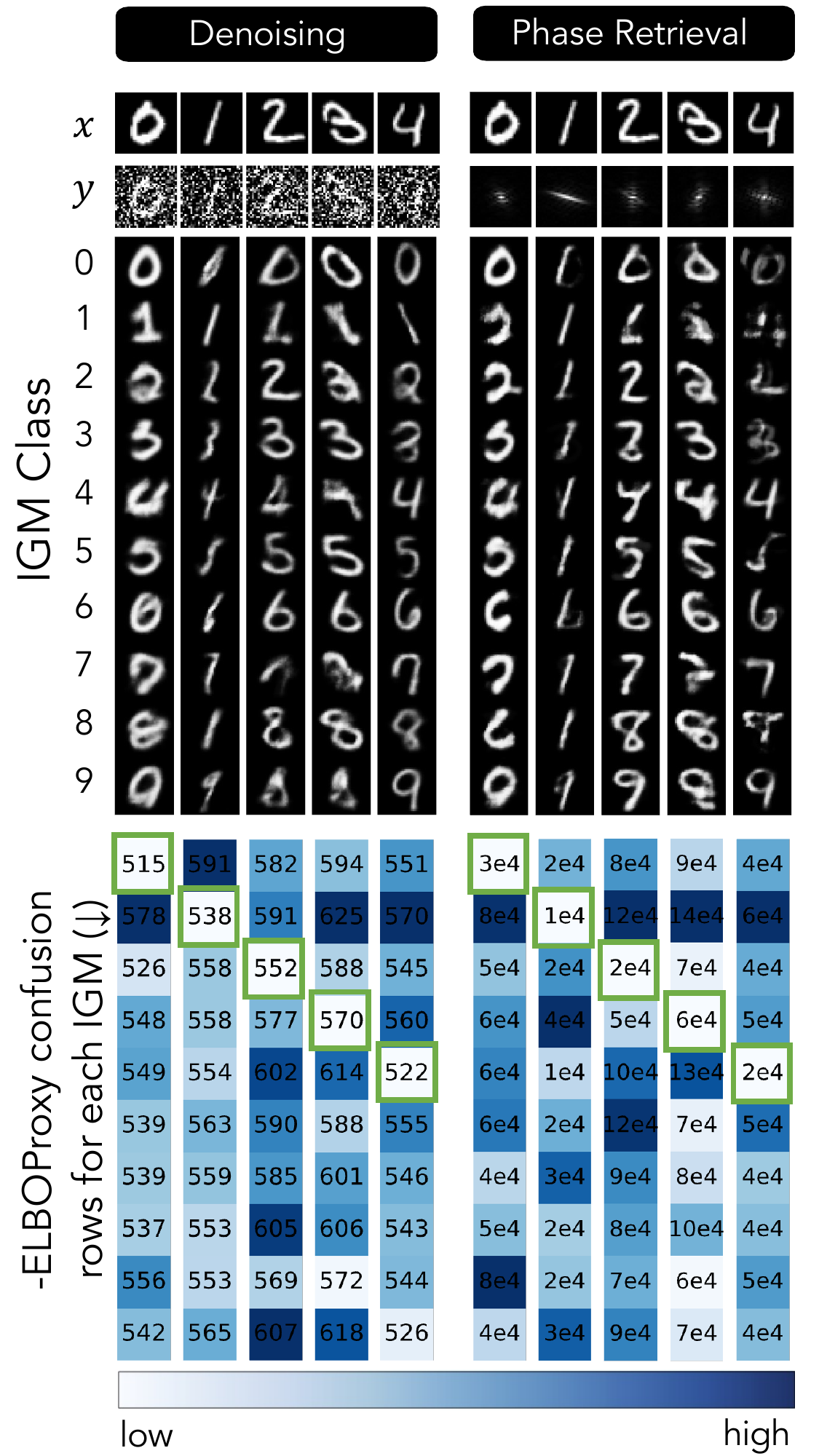}
    \caption{We consider two inverse problems: denoising and phase retrieval. Top: the two topmost rows correspond to the ground-truth image $x_c$ and the noisy measurements $y_c$. Center: in each column, we show the means of the distribution induced by the push-forward of $G_j$ and each latent distribution $z \sim q_{\phi_j}$ for $j \in \{0,\dots,9\}$. 
   Bottom: each column of the array corresponds to the negative $\ELBOProxy$ achieved by each model in reconstructing the images. Here, lower is better. Boxes highlighted in green correspond to the best negative $\ELBOProxy$ values in each column. In all these examples, the correct model was chosen.}
    \label{fig:MNIST_ELBO_exp}
\end{figure}

\begin{figure*}[ht]
    \centering
    \includegraphics[width=0.95\textwidth]{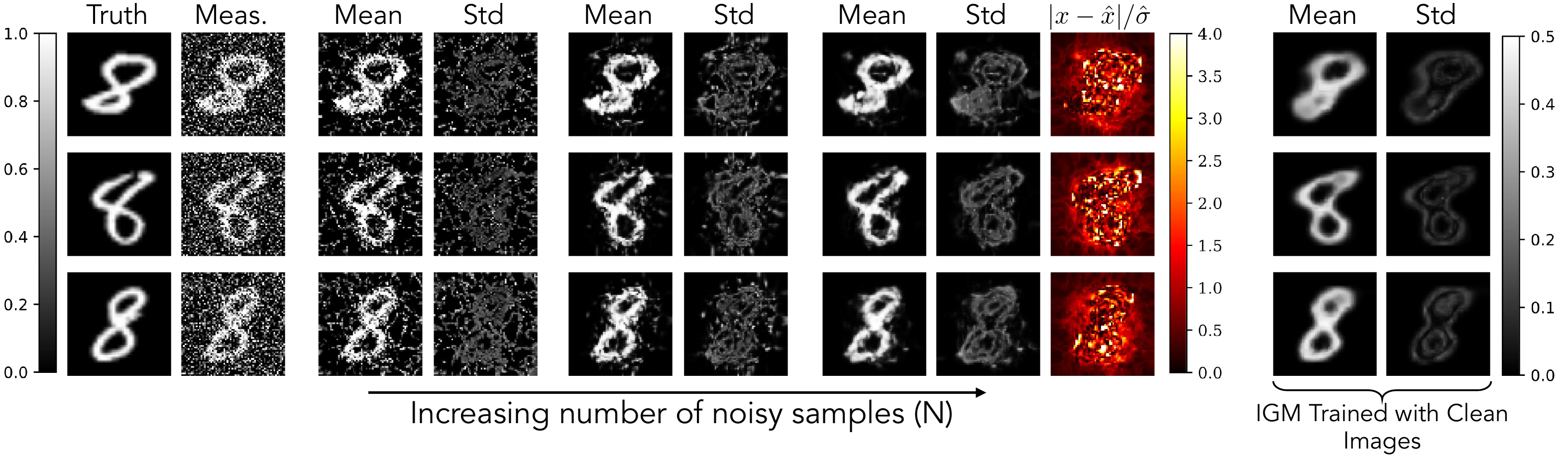}
    \caption{\textbf{Improvement with an increasing number of noisy observations.} We demonstrate our method of inferring an IGM to perform denoising for an increasing number of noisy MNIST images (5, 75, and 150 images from left to right). We showcase results on three randomly selected examples that appear in each collection of inverse problems. In each panel, we include the ground-truth, noisy measurements, mean of the posterior, and standard deviation of the posterior. We also include the residual error divided by the empirical standard deviation for $N = 150$. On the far right, we visualize reconstructions using an IGM trained on the full clean MNIST 8's class (6000 images). We observe that the mean reconstructions and standard deviations from our low-data IGMs become more similar to the full-data IGM with increasing data. Our residual errors are largely within 3 standard deviations.}
    \label{fig:MNIST_denoising}
\end {figure*}
\vspace{-3mm}
\subsection{Simultaneously solving many inverse problems}\label{sec:learning}

\begin{figure*}[h!]
    \centering
    \includegraphics[width=0.95\textwidth]{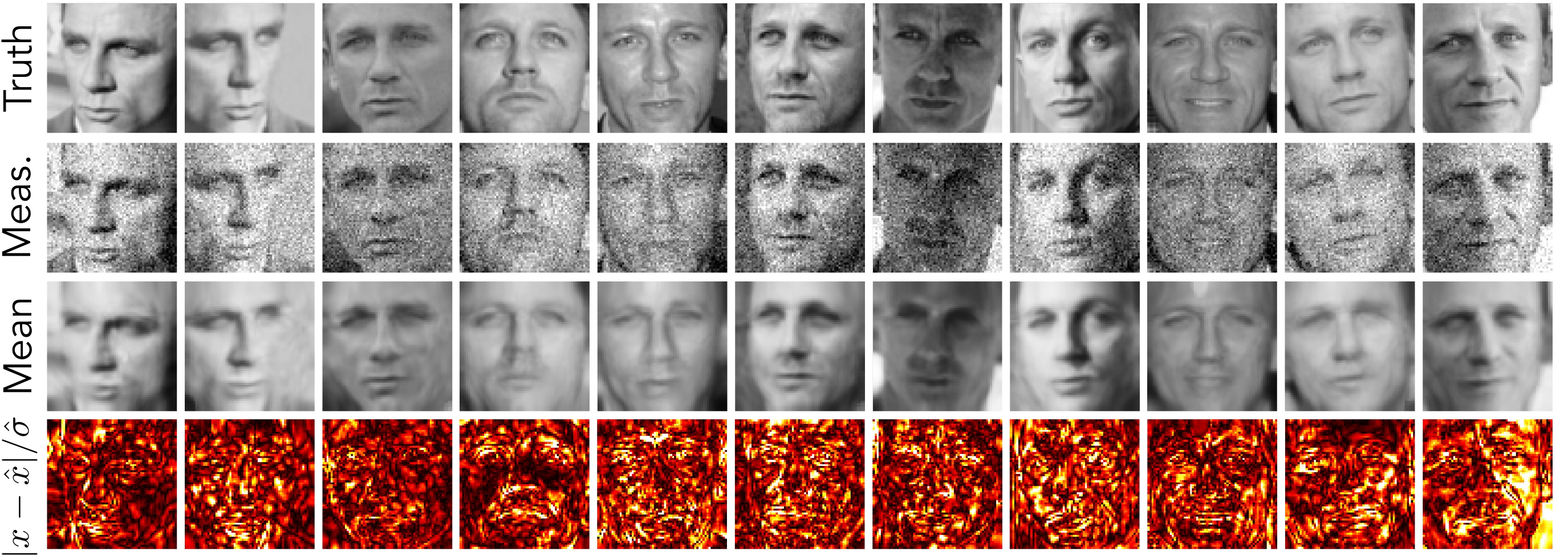}
    \caption{\textbf{Denoising 95 images of celebrity A.} We demonstrate our method described in \ref{sec:learning} using 95 noisy images of a celebrity. Here we show the underlying image (row 1), noisy measurements (row 2), mean reconstruction (row 3), and residual error (row 4) for a subset of the 95 different noisy images. Our reconstructions are much less noisy than the measurements and recover sharp features that are hard to discern in the noisy images. We visualize the residual error normalized by the empirical standard deviation, which indicates errors are largely within 3 standard deviations (refer to the colorbar in Fig.~\ref{fig:MNIST_denoising}). Note that no explicit spatial-domain prior/regularizer was used in denoising. }
    \label{fig:Bond_Denoising}
\end{figure*}



As the previous section illustrates, the $\ELBOProxy$ provides a good criterion for choosing an appropriate IGM from noisy measurements. 
Here, we consider the task of directly inferring the IGM from a collection of measurement examples $y^{(i)} = f^{(i)}(x^{(i)}) + \eta^{(i)}$ for $i \in [N]$, where the parameters are found by optimizing the $\ELBOProxy$. 
The key assumption we make is that common, low-dimensional structure is shared across the underlying images $\{x^{(i)}\}_{i=1}^N$. We propose to find a \textit{shared} generator $G_{\theta}$ with weights $\theta$ along with latent distributions $q_{\phi^{(i)}}$ that can be used to reconstruct the full posterior of each image $x^{(i)}$ from its corresponding measurement example $y^{(i)}$. This approach is illustrated in Fig.~\ref{fig:diagram}. Having the generator be shared across all images helps capture their common collective structure. Each forward model corruption, however, likely induces its own complicated image posteriors. Hence, we assign each measurement example $y^{(i)}$ its own latent distribution to capture the differences in their posteriors. Note that because we optimize a proxy of the ELBO, the inferred distribution may not necessarily be the true image posterior, but it still captures a distribution of images that fit to the observed measurements. 

\paragraph{Inference approach} 
More explicitly, given a collection of measurement examples $\{y^{(i)}\}_{i = 1}^N$, we jointly infer a generator $G_{\theta}$ and a set of variational distributions $\{q_{\phi^{(i)}}\}_{i = 1}^N$  by optimizing a Monte Carlo estimate of the $\ELBOProxy$ from Eq.~\eqref{eqref:ELBOProxy}, described by:
\begin{align}
    \{\hat{\theta}, \hat{\phi}^{(1)},\dots,\hat{\phi}^{(N)}\}&\in \argmax_{\theta, \{\phi^{(i)}\}_{i=1}^N} \mathcal{L}
\end{align} 
where 
\begin{align}
 \mathcal{L} =  & \frac{1}{N}\sum_{i = 1}^N \ELBOProxy(G_{\theta},q_{\phi^{(i)}}; y^{(i)}) + \log p(G_{\theta}). \label{eqref:learning-objective} 
\end{align}

In terms of choices for $\log p(G_{\theta})$, we can add additional regularization to promote particular properties of the IGM $G_{\theta}$, such as having a small Lipschitz constant. Here, we consider having sparse neural network weights as a form of regularization and use dropout \cite{srivastava2014dropout} during training to represent $\log p(G_{\theta})$.
 
 Once a generator $G_{\hat{\theta}}$ and variational parameters $\hat{\phi}^{(i)}$ have been inferred, we solve the $i$-th inverse problem by simply sampling $\hat{x}^{(i)} = G_{\hat{\theta}}(\hat{z}^{(i)})$ where $\hat{z}^{(i)} \sim q_{\hat{\phi}^{(i)}}(z^{(i)})$ or computing an average $\overline{x}^{(i)}=\frac{1}{T}\sum_{t=1}^T G_{\hat{\theta}}(\hat{z}_t^{(i)})$. Producing samples for each inverse problem can help visualize the range of uncertainty under the learned IGM $G_{\hat{\theta}}$, while the expected value of the distribution empirically provides clearer estimates with better metrics in terms of PSNR or MSE. We report PSNR outputs in our subsequent experiments and also visualize the standard deviation of our reconstructions.

\vspace{-2mm}
\section{Experimental Results}
\label{sec:results-joint}


We now consider solving a set of inverse problems via the framework described in \ref{sec:learning}. For each of these experiments, we use a multivariate Gaussian distribution to parameterize each of the posterior distributions $q_{\phi^{(i)}}$ and a Deep Decoder \cite{HH2018} with $6$ layers, $150$ channels in each layer, a latent size of $40$, and a dropout of $10^{-4}$ as the IGM. The multivariate Gaussian distributions are parameterized by means and covariance matrices $\{\mu^{(i)}, \Lambda^{(i)} = U_iU_i^T + \varepsilon I\}_{i = 1}^N$, where $\varepsilon I$ with $\varepsilon = 10^{-3}$ is added to the covariance matrix to help with stability of the optimization. We choose to parameterize the latent distributions using Gaussians for memory
considerations. Note that the same hyperparameters are used for all experiments demonstrating our proposed method.

In our experiments, we also compare to the following baseline methods: AmbientGAN \cite{bora2018ambientgan}, Deep Image Prior (DIP) \cite{ulyanov2018deep}, and regularized maximum likelihood using total variation (TV-RML). AmbientGAN is most similar to our setup, as it constructs an IGM directly from measurement examples; however, it doesn't aim to estimate image reconstruction posteriors, but instead aims to learn an IGM that samples from the full underlying prior. TV-RML uses explicit total variation regularization, while DIP uses an implicit convolutional neural network ``prior''. As we will show, all these baseline methods require fine-tuning hyperparameters to each set of measurements in order to produce their best results. All methods can also be applied to a variety of inverse problems, making them appropriate choices as baselines. 



\begin{figure*}
    \centering
    \includegraphics[width=0.9\textwidth]{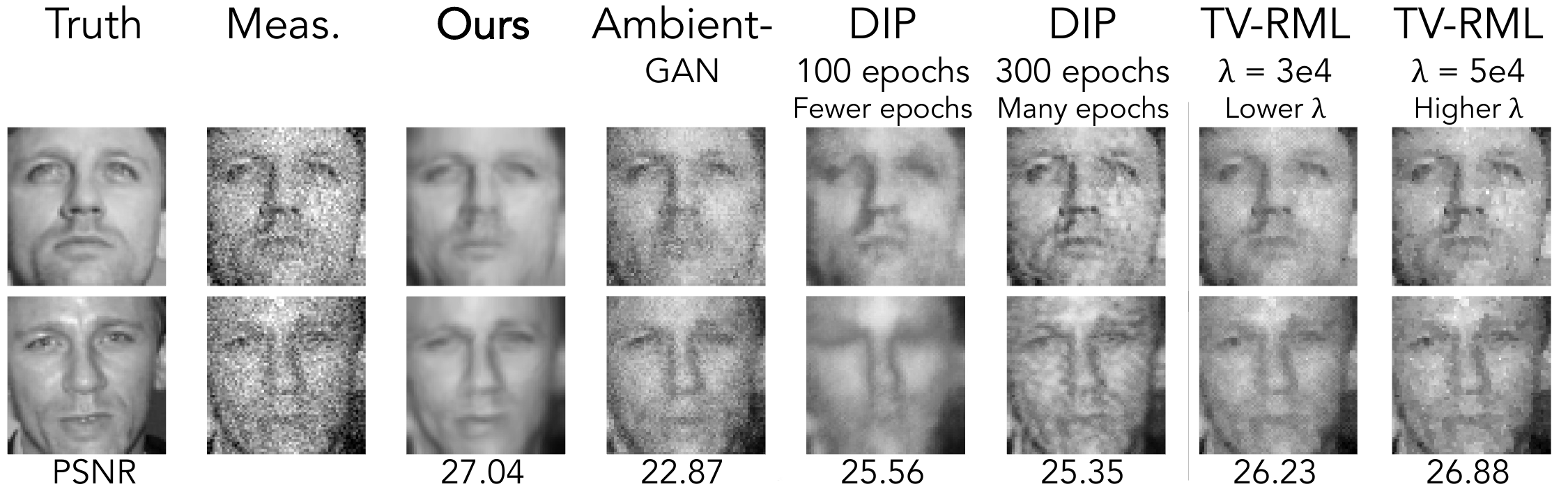}
    \caption{\textbf{Denoising baseline comparisons.} We compare to various baselines (AmbientGAN, Deep Image Prior (DIP), and regularized maximum likelihood using TV (TV-RML) with weight $\lambda$), and we report the average PSNR across all 95 reconstructions. We show both early stopping and full training results using DIP. Our method exhibits higher PSNR than all other baselines. We also include results for baselines that require fine-tuning to demonstrate sensitivity to subjective stopping conditions.}
    \label{fig:bondbaselines}
\end{figure*}

\begin{figure*}
    \centering
    \includegraphics[width=.95\textwidth]{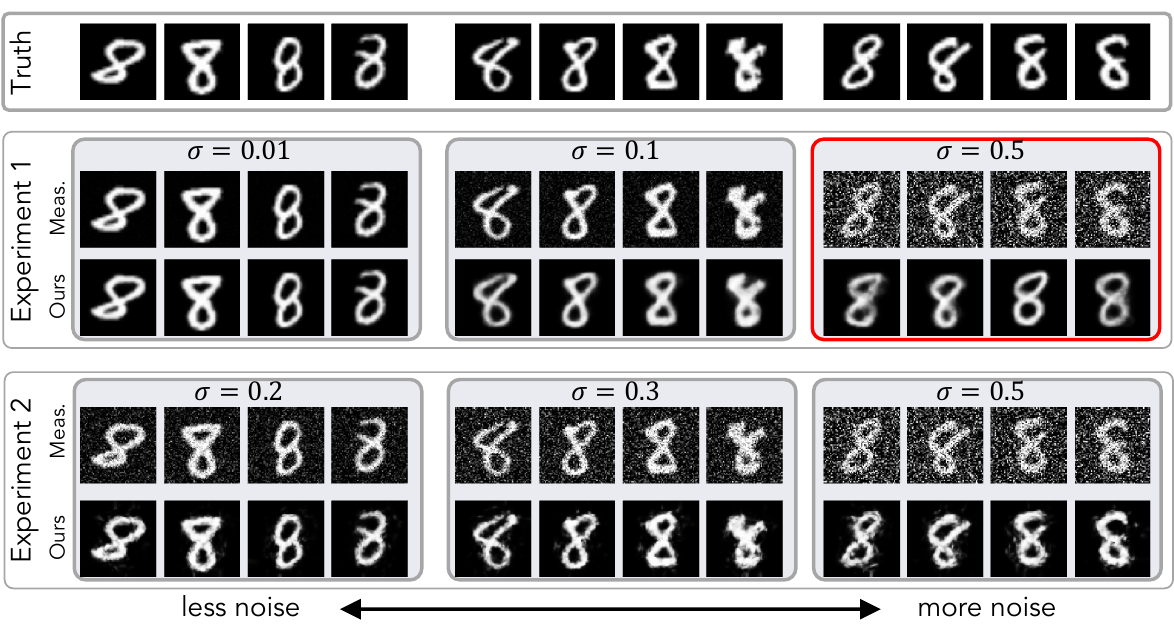}
    \caption{\textbf{Multi-noise denoising.} We demonstrate our method described in Section \ref{sec:learning} to perform denoising on measurement collections experiencing different noise levels. For each experiment, we use 75 measurement examples, which are defined by  $y^{(i)} = x^{(i)} + \eta$ where $\eta \in \{\eta_1, \eta_2, \eta_3\}$ and $\eta_i \sim \mathcal{N}(0, \sigma_i^2 I)$. In Experiment 1, we use noisy measurement examples that have additive noise with standard deviations of 0.01, 0.1, and 0.5. In Experiment 2, we use noisy measurement examples that have additive noise with standard deviations of 0.2, 0.3, and 0.5. We visualize the true underlying images, the measurement used for each experiment, and the mean of the image reconstruction posterior. Most of the reconstructions recover the primary features of the true image. However, in Experiment 1, the reconstructions of the low SNR measurements exhibit bias and do not match the true images well. This is likely because in Experiment 1, the high SNR measurements influence the inferred IGM more strongly than low SNR measurements, leading to biased reconstructions for the reconstructions highlighted in the red box.}
    \label{fig:multi_denoise_some}
\end{figure*}

\subsection{Denoising} \label{sec:denoising-exps}
We show results on denoising a collection of noisy images of 8's from the MNIST dataset in Fig.~\ref{fig:MNIST_denoising} and denoising a collection of noisy images of a single face from the PubFig \cite{kumar2009attribute} dataset in Fig.~\ref{fig:Bond_Denoising}. The measurements for both datasets are defined by $y = x + \eta$ where $\eta \sim \mathcal{N}(0,\sigma^2I)$ with an SNR of $\sim$-3 dB for the MNIST digits and an SNR of $\sim$15 dB for the faces. Our method is able to remove much of the added noise and recovers small scale features, even with only 10's of observations. As shown in Fig.~\ref{fig:MNIST_denoising}, the reconstructions achieved under the learned IGM improves as the number of independent observations increases. Our reconstructions also substantially outperform the baseline methods, as shown in Fig.~\ref{fig:bondbaselines}. Unlike DIP, our method does not overfit and does not require early stopping. Our method does not exhibit noisy artifacts like those seen in all baselines methods, despite such methods being fine-tuned. We show quantitative comparisons in Table~\ref{table:psnr}.

In Fig.~{\ref{fig:multi_denoise_some}} we show additional \textit{multi-noise} denoising experiments where we have $75$ noisy images, which have 3 different noise levels. More formally, $y^{(i)} = x^{(i)} + \eta$ where $\eta \in \{\eta_1, \eta_2, \eta_3\}$ and $\eta_i \sim \mathcal{N}(0, \sigma_i^2 I)$.  In Experiment 1, the noise levels have a wide range, and we use standard deviations of $\{\sigma_1, \sigma_2, \sigma_3\} = \{0.01, 0.1,  0.5\}$. In Experiment 2, the noise levels are much closer together, and we use standard deviations of $\{\sigma_1, \sigma_2, \sigma_3\} = \{0.2, 0.3,  0.5\}$. When the SNRs are similar (as in Experiment 2), the reconstructions match the true underlying images well. However, when the measurements have a wide range of SNRs (i.e., Experiment 1),  the reconstructions from low SNR measurements exhibit bias and poorly reconstruct the true underlying image, as shown in Fig.\ref{fig:multi_denoise_some}. This is likely because the high SNR measurements influence the inferred IGM more strongly than the low SNR measurements. The full set of results are available in the supplemental materials.

\begin{figure}
    \centering
    \includegraphics[width=0.45\textwidth]{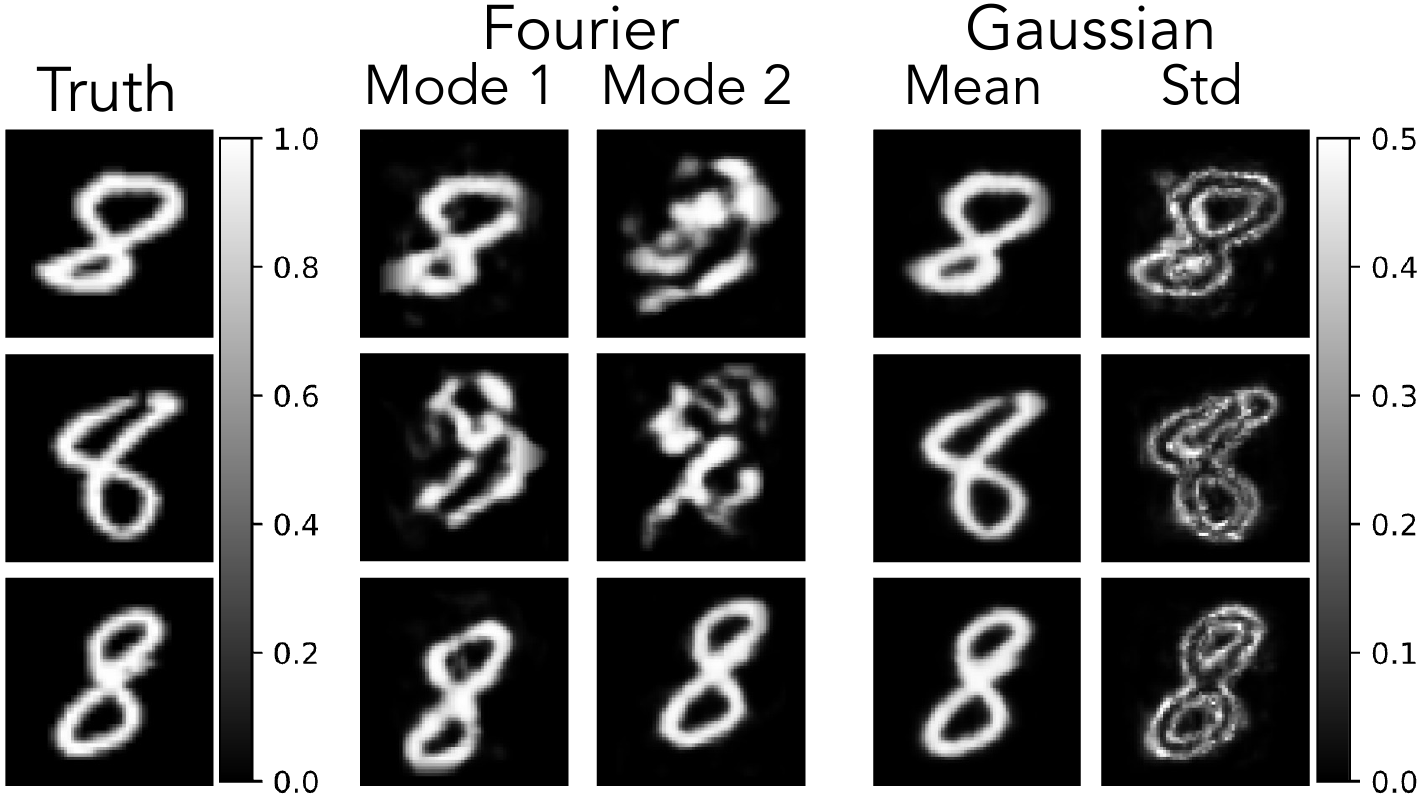}
    \caption{\textbf{Phase retrieval from MNIST 8's.} We demonstrate our method described in \ref{sec:learning} to perform phase retrieval on 150 images. For the Fourier phase retrieval setting, we show examples from the two observed modes of the posterior. For the Gaussian phase retrieval setting, we show the mean and standard deviation of our reconstructions. 
    }
    \label{fig:Mnist_phase_retrieval}
\end{figure}

\begin{figure}
    \centering
    \includegraphics[width=.45\textwidth]{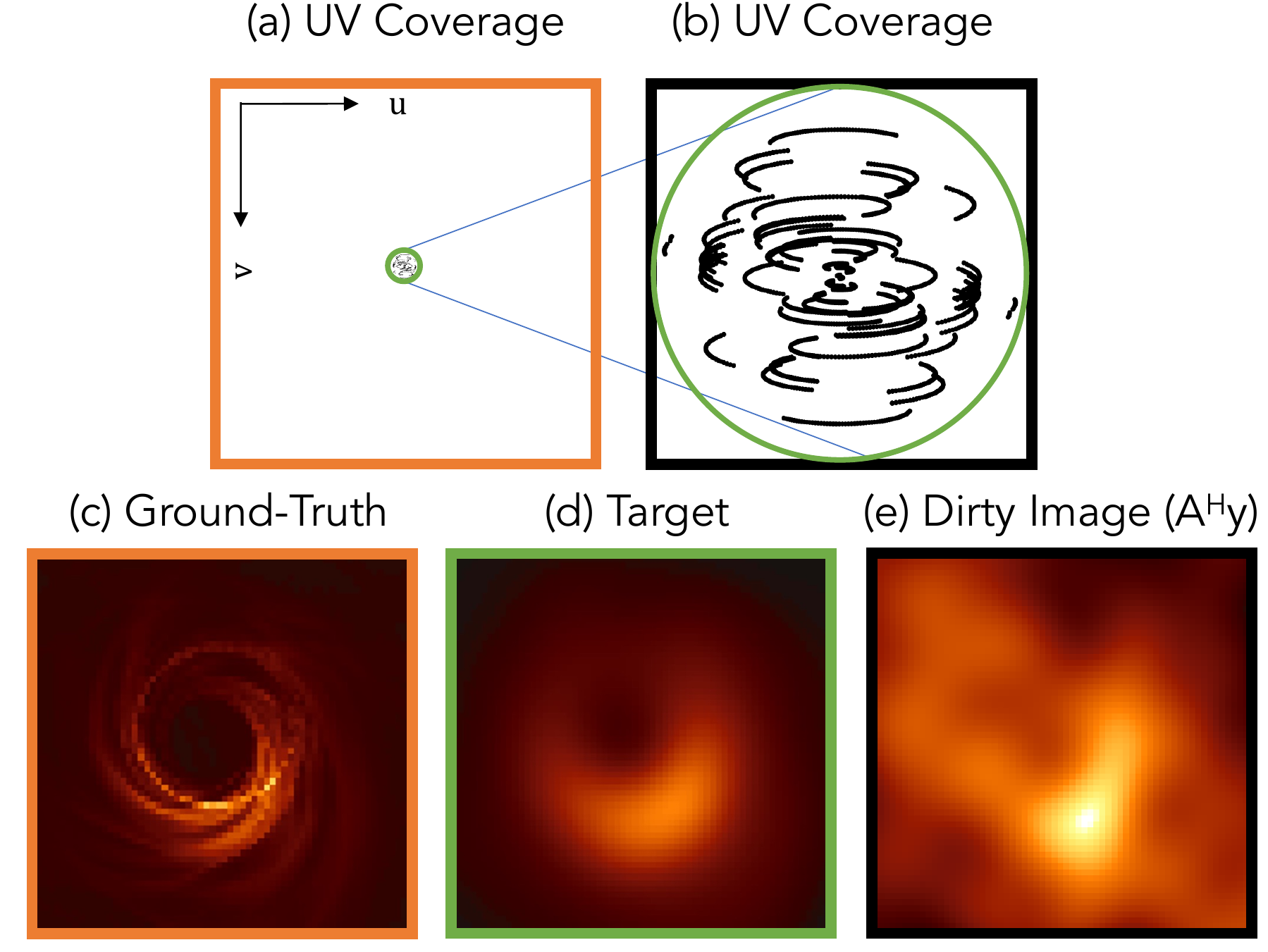}
    \caption{\textbf{Visualization of the intrinsic resolution of the EHT compressed sensing measurements.} The EHT measures sparse spatial frequencies of the image (i.e., components of the image's Fourier transform). In order to generate the underlying image (c), all frequencies in the entire domain of (a) must be used. Restricting spatial frequencies to the ones in (a) and (b)'s green circle generates the target (d), where (b) is a zoom in of (a).
    The EHT samples a subset of the interior of the green circle, indicated by the sparse black samples in (b). Naively recovering an image using only these frequencies results in the \textit{dirty image} (e), which is computed by $A^Hy$. The 2D spatial Fourier frequency coverage represented with $(u, v)$ positions is referred to as the UV coverage.} 
    \label{fig:uvcoverage}
\end{figure}

\vspace{-4mm}
\subsection{Phase retrieval} Here we consider solving non-convex inverse problems, and demonstrate our approach on phase retrieval. Our measurements are described by $y = |\mathcal{F}(x)| + \eta$ where $\mathcal{F}(\cdot)$ is a linear operator and $\eta \sim \mathcal{N}(0, \sigma^2 I)$. We consider two types of measurements, one where $\mathcal{F}(\cdot)$ is the Fourier transform and the other when $\mathcal{F}(\cdot)$ is an  $m \times n$ complex Gaussian matrix with $m = \lceil 0.1 n \rceil$. Since each measurement is the magnitude of complex linear measurements, there is an inherent phase ambiguity in the problem.
Additionally, flipping and spatial-shifts are possible reconstructions when performing Fourier phase retrieval. 
Due to the severe ill-posedness of the problem, representing this complicated posterior that includes all spatial shifts is challenging. Thus, we incorporate an envelope (i.e., a centered rectangular mask) as the final layer of $G_{\theta}$ to encourage the reconstruction to be centered. Nonetheless, flipping and shifts are still possible within this enveloped region. 

\begin{figure*}
    \centering
    \includegraphics[width=0.95\textwidth]{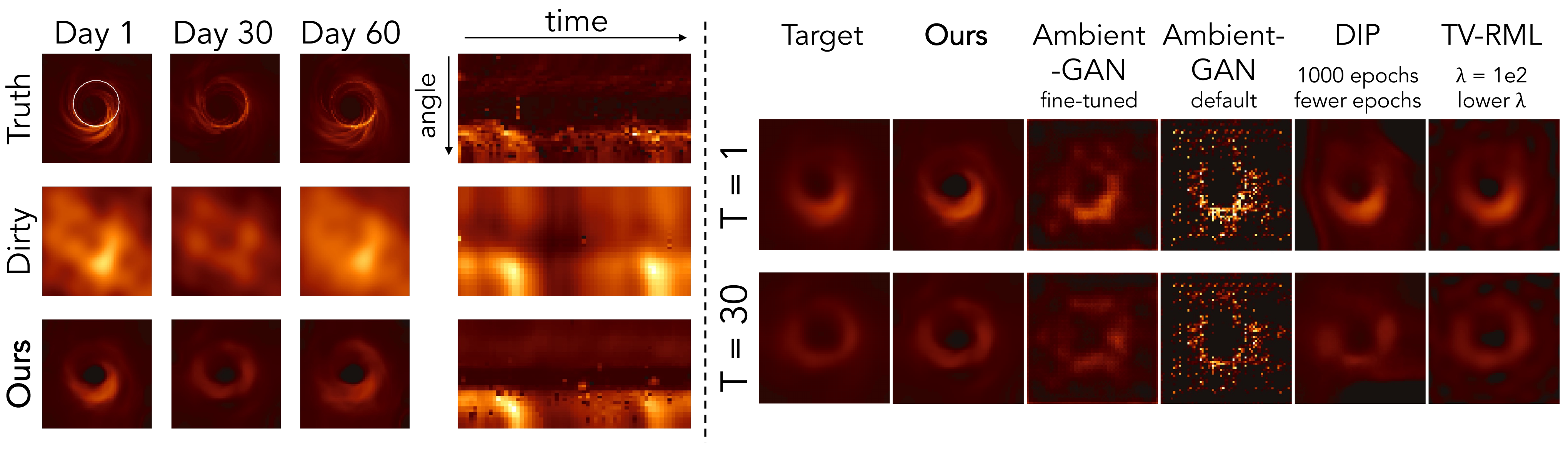}
    \caption{\textbf{Recovering a 60-day video of the M87 black hole with the EHT telescope array.} Left: We demonstrate our method described in \ref{sec:learning} using simulated measurements from 60 frames of an evolving black hole target with the forward model described in Fig.~\ref{fig:uvcoverage}. Here we show the underlying images, dirty images ($A^Hy$, see Fig.~\ref{fig:uvcoverage}), and mean reconstruction, respectively. Additionally, we show the unwrapped space × time image, which is taken along the overlaid white ring illustrated in the day 1 underlying image. The bright-spot's temporal trajectory of our reconstruction matches that of the truth.  Right: We compare our method to various baselines methods. Our results are much sharper and exhibit less artifacts than AmbientGAN and TV-RML with weight $\lambda$. Note that we include results using AmbientGAN with both default parameters and fine-tuned parameters. }
    \label{fig:m87_all}
\end{figure*}

\begin{table*}[h!]
\centering
\begin{adjustbox}{width=1\textwidth,center=\textwidth}
\begin{tabular}{l|l|l|c|ccccc}
\hline
     & &   &  & \multicolumn{5}{c}{Baselines}\\
    &  & Forward model & Ours & AmbGAN & DIP  & DIP  & TV  & TV  \\
               Dataset &  $N$ &   $y = f(x) + \eta$ &          &                fine-tuned &             fewer epochs &            many epochs &      lower $\lambda$ &     higher $\lambda$ \\
\hline
celeb. A &  95&  $y=x+\eta$ &                 \textbf{27.0} &                22.9 &                   25.6 &                   25.4 &               26.2 &               26.9 \\
celeb. B &  95&  $y=x+\eta$ &                 26.2             &    19.7 &                   24.4 &                   25.2 &                25.9 &                \textbf{26.4} \\
  MNIST 8's & 150&   $y=x+\eta$ &                  \textbf{21.1} &                18.0 &                    18.8 &                    13.3 &                16.4 &                18.2 \\
  \hline
  M87* (target) & 60&  $y=Ax+\eta$ &                 \textbf{29.3} &                  25.7 &                    29.0 &                    28.6 &                24.3 &                25.9 \\
\hline
\end{tabular} 
\end{adjustbox}
\caption{\textbf{Quantitative comparisons to baselines.} We show the mean PSNR between the reconstructions of different methods (ours, AmbientGAN \cite{bora2018ambientgan}, DIP \cite{ulyanov2018deep}, and TV-RML) with either the true underlying image or the target in the case of black hole compressed sensing. The highest PSNR for each set of measurements is in bold. We fix the hyperparameters of the baselines across each forward model, empirically selecting them for good performance on each type of problem. Note that the range in performance depends on the choice of hyperparameters. For denoising, the DIP results corresponds to 100 and 300 epochs for the fewer and many epochs baselines, respectively, and the lower and higher $\lambda$'s are 1e3 and 1e4, respectively. For black hole compressed sensing, the DIP results corresponds to 1000 and 3000 epochs for the fewer and many epochs baselines, respectively, and the lower and higher $\lambda$'s are 1e2 and 1e3, respectively. Visual examples of celeb B. are shown in the supplemental materials.}
\label{table:psnr}
\end{table*}

We show results from a set of $N = 150$ noisy phase retrieval measurements from the MNIST 8's class with a SNR of $\sim$52 dB. We consider three settings: 1) all measurements arise from a Gaussian measurement matrix, 2) all measurements arise from Fourier measurements, and 3) half of the measurements are Gaussian and the other half are Fourier. We show qualitative results for cases 1 and 2 in Fig.~\ref{fig:Mnist_phase_retrieval}. In the Gaussian case, we note that our mean reconstructions are nearly identical to the true digits and the standard deviations exhibit uncertainty in regions we would expect (e.g., around edges). In the Fourier case, our reconstructions have features similar to the digit $8$, but contain artifacts. These artifacts are only present in the Fourier case due to additional ambiguities, which lead to a more complex posterior \cite{SparsePRoverview}. We also show the average PSNR of our reconstructions for each measurement model in Table \ref{table:phase-retrieval-psnr}. For more details on this experiment, please see Section \ref{appx:pr-details} of the supplemental materials.

\begin{table}[h!]
\centering
\begin{adjustbox}{width=0.4\textwidth,center=0.5\textwidth}
\begin{tabular}{l|rrr}
\hline
& \multicolumn{3}{c}{Measurement Operator(s)}\\
{} &   Gaussian  &   Fourier  &   Both \\
\hline
Gaussian & 30.8 & --- & 30.2\\
Fourier & --- & 13.6 & 19.4\\
\hline
\end{tabular}
\end{adjustbox}
\caption{\textbf{Phase Retrieval PSNRs for different measurement operators.} Each column corresponds to the type of measurements that the method was given (a total of $N = 150$). In the case of ``Both'', $75$ Gaussian measurements and $75$ Fourier measurements were given. We then show the average PSNR for our reconstructions given the specific measurement operator (either Fourier or Gaussian). Note that when given both Gaussian and Fourier measurement examples our method exhibits a higher PSNR on the recovered images from Fourier examples than when given only Fourier measurements. Additionally, there is only a slight decrease in performance on the Gaussian measurement reconstructions as compared to given entirely Gaussian examples.}
\label{table:phase-retrieval-psnr}
\end{table}

\vspace{-4mm}
\subsection{Black hole imaging}

We consider a real-world inverse problem for which ground-truth data would be impossible to obtain. In particular, we consider a compressed sensing problem inspired by astronomical imaging of black holes with the Event Horizon Telescope (EHT): suppose we are given access to $N$ measurement examples of the form $y^{(i)} = A^{(i)}x^{(i)} + \eta^{(i)}$ where $A^{(i)}\in \mathbb{C}^{m \times n}$ is a low-rank compressed sensing matrix arising from interferometric telescope measurements and $\eta^{(i)}$ denotes noise with known properties (e.g., distributed as a zero-mean Gaussian with known covariance). The collection of images $\{x^{(i)}\}_{i=1}^N$ are snapshots of an evolving black hole target. This problem is ill-posed and requires the use of priors or regularizers to recover a reasonable image \cite{akiyama2019first4}. Moreover, it is impossible to directly acquire example images of black holes, so any pixel-level prior defined {\it a priori} will exhibit human bias. 
Recovering an image, or movie, of a black hole with as little human bias as possible is essential for both studying the astrophysics of black holes as well as testing fundamental physics \cite{m87I, sagAI}.
We show how our proposed method can be used to tackle this important problem.
In particular, we leverage knowledge that, although the black hole evolves, it will not change drastically from 
minute-to-minute or day-to-day. We study two black hole targets: the black hole at the center of the Messier 87 galaxy (M87$^*$) and the black hole at the center of the Milky Way galaxy -- Sagattarius A* (Sgr A$^*$). 

\paragraph{Imaging M87$^*$ using the current EHT array}
We first consider reconstructing the black hole at the center of the Messier 87 galaxy, which does not evolve noticeably within the timescale of a single day. The underlying images are from a simulated 60 frame video with a single frame for each day. We show results on $N = 60$ frames from an evolving black hole target with a diameter of $\sim$40 microarcseconds, as was identified as the diameter of M87 according to \cite{porth2019event, akiyama2019first5} in Fig.~\ref{fig:m87_all}. In particular, the measurements are given by $\{y^{(i)} = Ax^{(i)} + \eta^{(i)}\}_{i=1}^N$, where $x^{(i)}$ is the underlying image on day $i$, $A$ is the forward model that represents the telescope array, which is static across different days, and the noise $\eta^{(i)} \sim \mathcal{N}(0, \Sigma)$ has a covariance of  $\Sigma$, which is a diagonal matrix with realistic variances based on the telescope properties\footnote{We leave the more challenging atmospheric noise that appears in measurements for future work.}. Measurements are simulated from black hole images with a realistic flux of $1$ Jansky \cite{2022apJWong}. We also visualize a reference ``target'' image, which is the underlying image filtered with a low-pass filter that represents the maximum resolution achievable with the telescope array used to collect measurements -- in this case the EHT array consisting of 11 telescopes (see Fig.~\ref{fig:uvcoverage} and Section \ref{appx:EHT-arrays} in the supplemental materials). 

As seen in Fig.~\ref{fig:m87_all}, our method is not only able to reconstruct the large scale features of the underlying image without any aliasing artifacts, but also achieves a level of super-resolution (see Table~\ref{table:superres_table} in the supplemental materials). Our reconstructions also achieve higher super-resolution as compared to our baselines (i.e., AmbientGAN, TV-RML, and DPI) in Fig.~\ref{fig:m87_all} and do not exhibit artifacts evident in the reconstructions from these baselines. The two AmbientGAN settings were qualitatively chosen to show that the final result is sensitive to the choice of hyperparameters. The default AmbientGAN parameters produce poor results, and even with fine-tuning to best fit the underlying images (i.e., cheating with knowledge of the ground-truth), the results still exhibit substantial artifacts. We outperform  the baselines in terms of PSNR when compared to the target image (see Table~\ref{table:psnr}). Our results demonstrate that we are able to capture the varying temporal structure of the black hole, rather than just recovering a static image. It is important to note that there is no explicit temporal regularization introduced; the temporal smoothness is implicitly inferred by the constructed IGM.  



\begin{figure}
    \centering
    \includegraphics[width=0.44\textwidth]{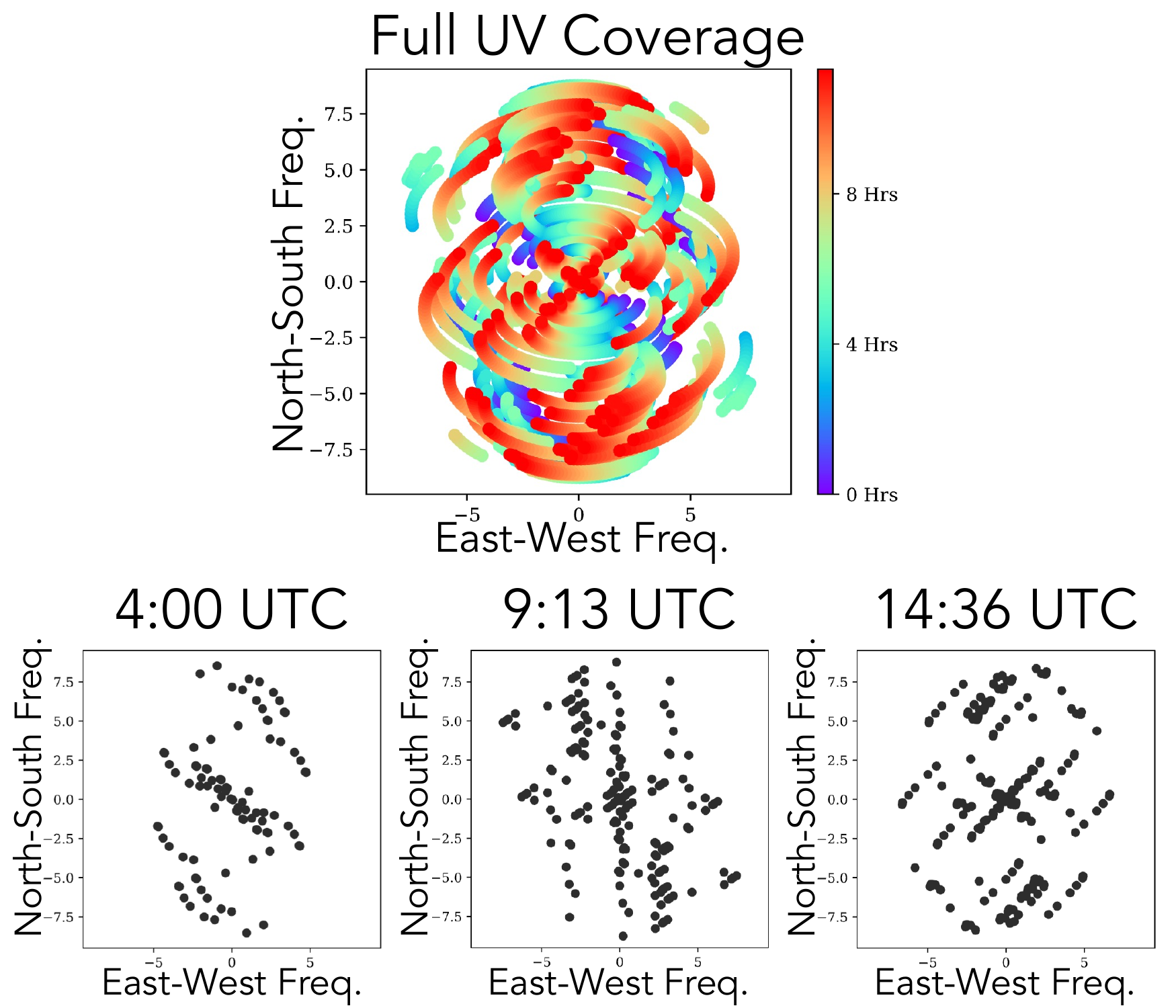}
    \caption{\textbf{Visualization of the EHT forward model of Sgr A$^*$ over time.} The spatial frequency space coordinates of Sgr A$^*$ from an ngEHT \cite{raymond2021evaluation} array measured over a course of one night is shown on the top. Snapshots of the coordinates that are measured at different times are shown on the bottom. Note that the location and sparsity of the measurement coordinates change over time.}
    \label{fig:sagA_coverage}
\end{figure}

\begin{figure*}
    \centering
    \includegraphics[width=0.87\textwidth]{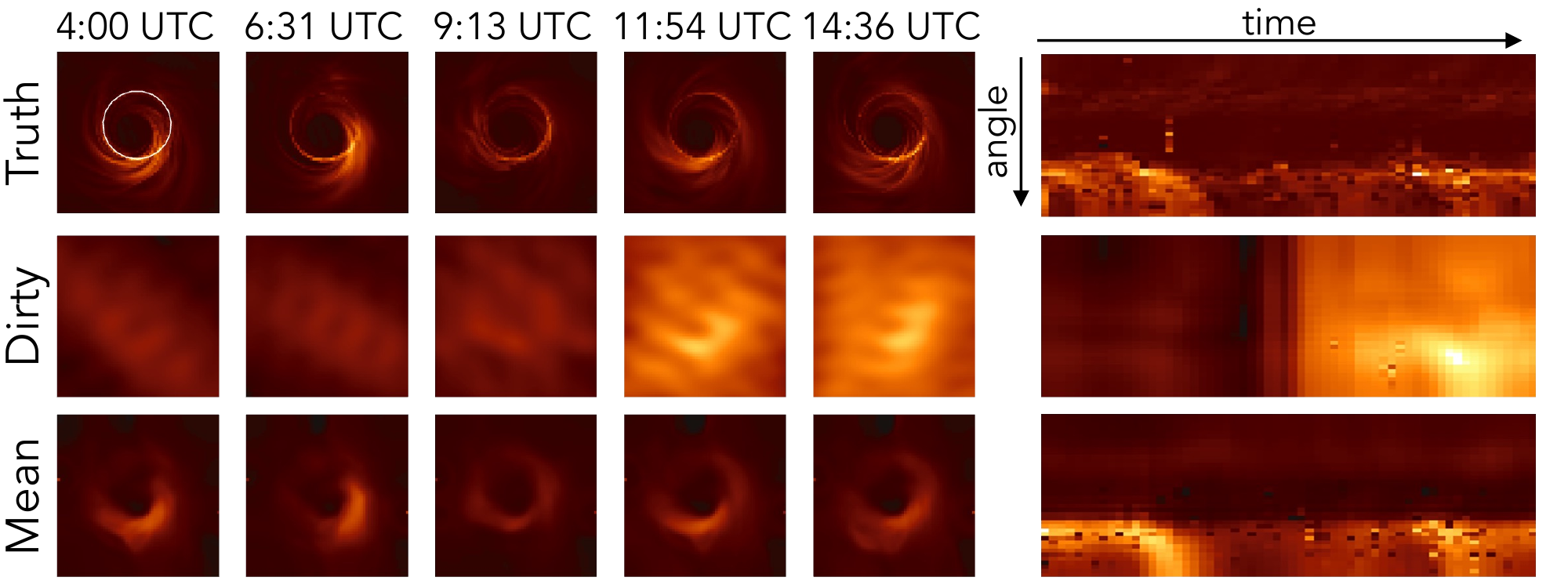}
    \caption{\textbf{Recovering a video of the Sgr A* black hole with a futuristic telescope array over the course of a night.} We demonstrate our method described in \ref{sec:learning} to perform video reconstruction on 60 frames of a video, where the measurements for each frame are generated by different forward models given by imaging Sgr A$^*$. Here we show the underlying image, the dirty image ($A^Hy$), and the mean of the reconstructed posterior. Additionally, we show the unwrapped space × time image, which is taken along the white ring illustrated in the 4:00 UTC underlying image. The bright-spot's temporal trajectory of our reconstruction matches that of the truth. The measurement noise is consistent with a black hole having a flux of 2 Janskys.}
    \label{fig:sagA}
\end{figure*}


\paragraph{Imaging Sgr A$^*$ from multiple forward models}\label{sec:results-joint-multi-fwd}

The framework we introduce can also be applied to situations in which the measurements themselves are induced by different forward models. In particular, the measurements $\{y^{(i)} = f^{(i)}(x^{(i)}) + \eta^{(i)}\}_{i=1}^N$ are given by an underlying image $x^{(i)}$, a forward model $f^{(i)}$ that is specific to that observation, and noise $\eta^{(i)}$ with known properties. 

As an illustrative example, we consider the problem of reconstructing a video of the black hole at the center of the Milky Way -- Sagittarius A$^*$ (Sgr A$^*$). Unlike M87$^*$, Sgr A$^*$ evolves on the order of minutes. 
Therefore, we can only consider that the black hole is static for only a short time when only a subset of telescopes are able to observe the black hole. This results in a different measurement forward model for each frame of the black hole ``movie'' \cite{akiyama2022firstIII}. In particular, the measurements are given by $\{y^{(i)} = A^{(i)}x^{(i)} + \eta^{(i)}\}_{i=1}^N$, where $x^{(i)}$ is the underlying image at time $i$, $A^{(i)}$ is the forward model that incorporates the telescope configuration at that time, and $\eta^{(i)} \sim \mathcal{N}(0, \Sigma^{(i)})$ is noise where $\Sigma^{(i)}$ is a diagonal matrix with realistic standard deviations derived from the telescopes' physical properties. The measurement noise is consistent with a black hole with a flux of 2 Janskys. The measurement operator is illustrated in Fig.~\ref{fig:sagA_coverage}. 

We show examples of reconstructing 60 frames of Sgr A$^*$ with a diameter of $\sim$50 microarcseconds using measurements simulated from a proposed future next-generation EHT (ngEHT) \cite{raymond2021evaluation} array, which consists of 23 telescopes. These results are shown in Fig.~\ref{fig:sagA}. Our reconstructions remove much of the aliasing artifacts evident in the dirty images and reconstructs the primary features of the underlying image without any form of temporal regularization. These results have high fidelity especially considering that the measurements are very sparsely sampled.
Although these results come from simulated measurements that do not account for all forms of noise we expect to encounter, the high-quality movie reconstructions obtained without the use of a spatio-temporal prior show great promise towards scientific discovery that could come from a future telescope array paired with our proposed reconstruction approach.

\vspace{-2mm}
\section{Theory for Linear IGMs}

\label{sec:theory}

We now introduce theoretical results on the inferred IGM for linear inverse problems. Specifically, we consider the case when the IGM $G_{\theta}$ is linear and the latent variational distributions are Gaussian. The goal of this section is to develop intuition for the inferred IGM in a simpler setting. While our results may not generalize to non-linear generators parameterized by deep neural networks, our results aim to provide an initial understanding on the inferred generator.


More concretely, suppose we are given $N$ measurement examples of noisy linear measurements of the form \begin{align*}
    y^{(i)} := Ax^{(i)} + \eta^{(i)},\ \eta^{(i)} \sim \mathcal{N}(0,\sigma^2 I)
\end{align*} where $A \in \R^{m \times n}$ with $m \leqslant n$ and $x^{(i)} \in \R^n$. We aim to infer $G_{\theta} \in \R^{n \times k}$ with $k < n$ and latent distributions $q_{\phi^{(i)}} = \mathcal{N}(\mu_i, U_i U_i^T)$ where $\phi^{(i)} = \{\mu_i, U_i\}$, $\mu_i \in \R^k$, and $U_i \in \R^{k \times k}$ by minimizing the negative ELBOProxy:\begin{align}
     \mathcal{L}(G_{\theta},\{\phi^{(i)}\}) & := \frac{1}{N}\sum_{i=1}^N \E_{z \sim q_{\phi^{(i)}(z)}}[-\log p(y^{(i)} |G_{\theta}(z)) \nonumber \\
     & - \log p_Z(z | G_{\theta}) + \log q_{\phi^{(i)}}(z)]. \label{eq:linear_theory_loss}
\end{align}

We characterize the generator $G_*$ and latent parameters $\phi_*^{(i)} = \{\mu_i^*, U_i^*\}$ that are stationary points of \eqref{eq:linear_theory_loss}. The result is proven in Section \ref{appx:supplementary-theory}: \begin{thm} \label{thm:main-estimators}
Fix $\sigma > 0$ and let $y^{(i)} = Ax^{(i)} + \eta^{(i)} \in \R^n$ for $i \in [N]$ where $\eta^{(i)} \sim \mathcal{N}(0,\sigma^2 I)$ and $A \in \R^{m \times n}$ has $k \leqslant \mathrm{rank}(A) = m\leqslant n$. Define the sample covariance of the measurements $Y_N := \frac{1}{N}\sum_{i=1}^N y^{(i)}(y^{(i)})^T$. Then, with probability $1$,  $G_*$, $\mu_i^*$, and $U_i^*$ that satisfy \begin{align*} AG_* & := E_k(L_k - \sigma^2I)^{1/2}R, \\
\mu_i^* & := (G_*^TA^TAG_* + \sigma^2 I)^{-1}G_*^TA^Ty^{(i)},\ \text{and} \\ U_i^* & := \sigma (G_*^TA^TAG_* + \sigma^2I)^{-1/2}\tilde{R},
\end{align*} for $i \in [N]$ are stationary points of the objective \eqref{eq:linear_theory_loss} over $\R^{n \times k} \times (\R^k \times \mathrm{GL}(k,\R))^N$ where $\mathrm{GL}(k,\R)$ denotes the set of real invertible $k \times k$ matrices. Here, $E_k \in \R^{m \times k}$ denotes the matrix whose columns correspond to the top-$\min\{\mathrm{rank}(Y_N),k\}$ eigenvectors of $Y_N$, $L_k \in \R^{k \times k}$ contains the corresponding eigenvalues of $Y_N$, and $R, \tilde{R} \in \R^{k \times k}$ are arbitrary orthogonal matrices. If $\mathrm{rank}(Y_N) < k$, then for $\mathrm{rank}(Y_N) < i\leqslant k$, the $i$-th column of $E_k$ can be arbitrary and $L_{k,ii} = \sigma^2$.
\end{thm}

The Theorem establishes the precise form of stationary points of the objective \eqref{eq:linear_theory_loss}. In particular, it shows that this inferred IGM performs dimensionality reduction akin to PCA \cite{TippingBishop} on the collection of measurement examples. To gain further intuition about the Theorem, in Section \ref{appx:supplementary-theory-denoising-ex} of the supplemental materials, we analyze this result in a context where the underlying images we wish to reconstruct explicitly lie in a low-dimensional subspace. We show that, in that setting, our estimator returns an approximation of the solution found via MAP estimation, which can only be computed with complete prior knowledge of the underlying image structure. While the Theorem focused on linear IGMs, it would be interesting to theoretically analyze non-linear IGMs parameterized by deep networks. We leave this for future work.
\vspace{-2mm}
\section{Conclusion}
In this work we showcased how one can solve a set of inverse problems without a pre-defined IGM (e.g., a traditional spatial image prior) by leveraging common structure present across a collection of diverse underlying images. We demonstrated that even with a small set of corrupted measurements, one can jointly solve these inverse problems by directly inferring
an IGM that maximizes a proxy of the ELBO. We demonstrate our method on a number of convex and non-convex imaging problems, including the challenging problem of black hole video reconstruction from interferometric measurements.
Overall, our work showcases the possibility of solving inverse problems in a completely unsupervised fashion, free from significant human bias typical of ill-posed image reconstruction. We believe our approach can aid in automatic discovery of novel structure from scientific measurements, potentially paving the way to new avenues of exploration.

\bibliographystyle{plain}
\bibliography{references.bib}

\newpage
\tableofcontents

\section{Experimental details}
\label{appx:supplementary-exp-details}

 \paragraph{Model selection} We parameterize the latent variational distributions $q_{\phi}(z)$ with Normalizing Flows $l_{\phi}$. In particular, suppose our latent variables $z = l_{\phi}(s)$ where $s \sim \mathcal{N}(0,I)$. Then, \begin{align*}
     \log q_{\phi}(z) = \log \pi (s) - \log\left|\det\frac{dl_{\phi}(s)}{ds}\right|
 \end{align*} where $\pi(\cdot)$ denotes the likelihood of the base distribution of $s$, which in this case is a Gaussian. For our prior on $z$ given $G_{\theta}$, we use an isotropic Gaussian prior. Then the $\ELBOProxy$ equation becomes \begin{align*}
    \ELBOProxy(G_{\theta}, q_{\phi}; y) & := \E_{z \sim l_{\phi}(z)} [\log p(y|G(z)) \\
    & + \log p(z|G) - \log q_{\phi}(z)] \\
    & = \E_{s \sim \mathcal{N}(0,I)} \Big[\log p(y|G_{\theta}(l_{\phi}(s))) \\
    & + \frac{1}{2}\|l_{\phi}(s)\|^2_2 - \log \pi (s) \\
    & + \log\left|\det\frac{dl_{\phi}(s)}{ds}\right|\Big].
 \end{align*} Note that the expectation over $\log\pi(s)$ is constant with respect to the parameters $\phi$, so we only optimize the remaining terms.



\section{Inference Details}

\paragraph{Optimization details for model selection} Given a generator $G_c$ for a class $c$, we optimized the $\ELBOProxy$ with respect to the parameters $\phi$ of the Normalizing Flow. For denoising, we optimized for $5000$ epochs with Adam \citePhys{kingma2014adamSupp} and a learning rate of $10^{-4}$ and set $T = 200$. For phase retrieval, we optimized for $7500$ epochs with Adam and a learning rate of $10^{-4}$ and set $T = 1000$. The entire experiment took approximately $10/15$ hours for denoising/phase retrieval, respectively, on a single NVIDIA V100 GPU. When optimizing the flow $h_{\phi}$ for a particular generator $G_c$, we record the parameters with the best $\ELBOProxy$ value. This is the final value recorded in the figure. To generate the mean and standard deviation, we used $100$ samples.

\paragraph{Network details} The generator $G_c$ corresponds to the decoder part of a VAE \citePhys{KingmaWelling14Supp} trained on the training set of class $c$ of MNIST. We used a convolutional architecture, where the encoder and decoder architecture contained $4$ convolutional layers. There are also linear, fully-connected layers in between to learn the latent representation, which had dimension $16.$ The Normalizing Flows used the RealNVP \citePhys{dinh2016densitySupp} architecture. In both problems, the network had $16$ affine coupling layers with linear layers, Leaky ReLU activation functions, and batch normalization. After each affine-coupling block, the dimensions were randomly permuted.

\subsection{Inferring the IGM} 

\paragraph{Variational distribution} We parameterize the posterior using Gaussian distributions due to memory constraints; the memory of the whole setup with Normalizing Flows using 75 images is around 5 times larger than the setup using a Gaussian variational family. However, the Gaussian variational family is likely not as good of a variational distribution, so there will be a wider margin between $p(y|G)$ and the ELBO. Additionally, the Gaussian variational family will not be as good at capturing multimodal distributions. Note that we do not notice overfitting so there is no need for early stopping.

\paragraph{IGM network details} We use a Deep Decoder \citePhys{HH2018Supp} generator architecture as the IGM with a dropout of $10^{-4}$. The Deep Decoder has 6 layers with 150 channels each, and has a latent size of 40. There is a sigmoid activation as the final layer of the generator.

\paragraph{Optimization details} We train for 99,000 epochs using the Adam \citePhys{kingma2014adamSupp} optimizer with a learning rate of $10^{-4}$, which takes around 2 hours for 5 images, 17 hours for 75 images, and 33 hours for 150 images on a single NVIDIA RTX A6000 GPU. The memory usage is around 11000MiB for 5 images, 24000MiB for 75 images, and
39000MiB for 150 images. We use a batch size of 20 samples per measurement $i$, and we perform batch gradient descent based on all measurements $i \in \{1, \dots, N\}$.

\section{Event Horizon Telescope Array} \label{appx:EHT-arrays}
\subsection{Array used for the M87$^*$ results}

\begin{tabular}{lrrr}
NAME & X & Y & Z\\
\hline
PDB & 4523998 & 468045 & 4460310\\
PV & 5088968 & -301682 & 3825016\\
SMT & -1828796 & -5054407 & 3427865\\
SMA & -5464523 & -2493147 & 2150612\\
LMT & -768714 & -5988542 & 2063276\\
ALMA & 2225061 & -5440057 & -2481681\\
APEX & 2225040 & -5441198 & -2479303\\
JCMT & -5464585 & -2493001 & 2150654\\
CARMA & -2397431 & -4482019 & 3843524\\
KP & -1995679 & -5037318 & 3357328\\
GLT & 1500692 & -1191735 & 6066409\\
\end{tabular}

\subsection{Array used for the Sgr A$^*$ results}

\begin{tabular}{lrrr}
NAME & X & Y & Z\\
\hline
ALMA & 2225061 & -5440057 & -2481681\\
APEX & 2225040 & -5441198 & -2479303\\
GLT & 1500692 & -1191735 & 6066409\\
JCMT & -5464585 & -2493001 & 2150654\\
KP & -1995679 & -5037318 & 3357328\\
LMT & -768714 & -5988542 & 2063276\\
SMA & -5464523 & -2493147 & 2150612\\
SMT & -1828796 & -5054407 & 3427865\\
BAJA & -2352576 & -4940331 & 3271508\\
BAR & -2363000 & -4445000 & 3907000\\
CAT & 1569000 & -4559000 & -4163000\\
CNI & 5311000 & -1725000 & 3075000\\
GAM & 5648000 & 1619000 & -2477000\\
GARS & 1538000 & -2462000 & -5659000\\
HAY & 1521000 & -4417000 & 4327000\\
NZ & -4540000 & 719000 & -4409000\\
OVRO & -2409598 & -4478348 & 3838607\\
SGO & 1832000 & -5034000 & -3455000\\
CAS & 1373000 & -3399000 & -5201000\\
LLA & 2325338 & -5341460 & -2599661\\
PIKE & -1285000 & -4797000 & 3994000\\
PV & 5088968 & -301682 & 3825016\\
\end{tabular}

\section{Supplemental Theoretical Results}
\label{appx:supplementary-theory}

We present proofs of some of the theoretical results that appeared in the main body of the paper, along with new results that were not discussed. Throughout the proofs, we let $\|\cdot\|_2$, $\|\cdot\|$, and $\|\cdot\|_F$ denote the Euclidean norm for vectors, the spectral norm for matrices, and the Frobenius norm for matrices, respectively. For a matrix $M \in \R^{m \times n}$, let $\sigma_i(M)$ and $\lambda_i(M)$ denote the $i$-th singular value and eigenvalue, respectively (ordered in decreasing order). Let $I_n$ denote the $n \times n$ identity matrix. We may use $I$ for simplicity when the size of the matrix is clear from context. To avoid cumbersome notation, we sometimes use the shorthand notation $\E_z[\cdot]$ for $\E_{z \sim \mathcal{D}}[\cdot]$. 

\subsection{Proof of Theorem \ref{thm:main-estimators}}

We first focus on proving Theorem \ref{thm:main-estimators}. To show this, we first establish that the $\ELBOProxy$ loss can be evaluated in closed-form:
\begin{prop} \label{prop:loss_closed_form}
Fix $\sigma > 0$ and suppose $y^{(i)} := Ax^{(i)} + \eta^{(i)},\ \eta^{(i)} \sim \mathcal{N}(0,\sigma^2 I)$ for $i \in [N]$ where $A \in \R^{m \times n}$. Then the loss $\Lcal : \R^{n \times k} \times (\R^k \times \R^{k \times k})^N \rightarrow \R$ in \eqref{eq:linear_theory_loss} can be written as \begin{align*}
    \frac{1}{N}& \sum_{i=1}^N\Big[\frac{1}{2\sigma^2}\left(\|y^{(i)}\|_2^2 + \|AG_{\theta}U_i\|_F^2 + \|AG_{\theta}\mu_i\|_2^2\right) \\
   &- \frac{1}{\sigma^2} \langle y^{(i)}, AG_{\theta} \mu_i\rangle   + \frac{1}{2}(\|U_i\|_F^2 + \|\mu_i\|_2^2) \\
   & - \log\det(U_i)\Big]  -\frac{k}{2}- \frac{k}{2}\log (2\pi).
\end{align*}
\end{prop}

\begin{proof}[Proof of Proposition \ref{prop:loss_closed_form}]

Since the loss is separable in $N$ and expectation is linear, we can consider computing the loss for a single $i \in [N]$. Dropping the indices for notational simplicity, consider a fixed $y = Ax + \eta \in \R^m$ where $\eta \sim \Ncal(0,\sigma^2I_m)$. We aim to calculate \begin{align*}
     \E_{v}\left[- \log p(y | G_{\theta}(v)) - \log p(v |G_{\theta}) + \log q_{\phi}(v)\right]
\end{align*} where $p(y|G_{\theta}(v))$ is Gaussian with variance $\sigma^2$, $p(v|G_{\theta})$ is isotropic Gaussian, and $v \sim q_{\phi}=\Ncal(\mu, UU^T)$. Note that $v \sim \Ncal(\mu, UU^T)$ is equivalent to $v = Uz + \mu$ for $z \sim \mathcal{N}(0,I_k)$. We can calculate the individual terms in the loss as follows. We first collect some useful results, whose proofs are deferred until after this proof is complete:
\begin{lemma} \label{lem:helpful-expectations}
For any matrix $U \in \R^{d \times k}$, we have $\E_z\|Uz\|_2^2 = \|U\|_F^2$ where $z \sim \mathcal{N}(0,I_k)$. Moreover, if $UU^T$ is invertible, then $\E_z[z^TU^T(UU^T)^{-1}Uz] = d$.
\end{lemma}

Now, for the first term, since the noise is Gaussian with variance $\sigma^2$, it is equal to \begin{align*}
     \E_z\frac{1}{2\sigma^2}\|y & - AG_{\theta}(Uz + \mu)\|_2^2  = \frac{1}{2\sigma^2}\|y\|_2^2 \\
     &+ \E_z\frac{1}{2\sigma^2}\|AG_{\theta}(Uz + \mu)\|_2^2 \\
     & - \frac{1}{\sigma^2}\E_z\langle y, AG_{\theta}(Uz + \mu)\rangle \\
    & = \frac{1}{2\sigma^2}\left(\|y\|_2^2 + \|AG_{\theta}U\|_F^2 + \|AG_{\theta}\mu\|_2^2\right) \\
    & - \frac{1}{\sigma^2} \langle y, AG_{\theta} \mu\rangle
\end{align*} where we used the first part of Lemma \ref{lem:helpful-expectations} in the third equality. For the prior term, we can calculate using Lemma \ref{lem:helpful-expectations} again

 \begin{align*}
    \E_{v }\left[- \log p(v | G_{\theta})\right] & = \E_z\frac{1}{2}\|Uz + \mu\|_2^2 \\
    & = \frac{1}{2}(\|U\|_F^2 + \|\mu\|_2^2).
\end{align*} Finally, we can calculate the entropy term as

 \begin{align*}
    \E_{v}& \left[\log q_{\phi}(v)\right] = - \frac{1}{2}\E_{v }\left[(v - \mu)^T(UU^T)^{-1}(v-\mu)\right] \\
    & - \frac{k}{2}\log (2\pi) - \frac{1}{2}\log|\det(UU^T)| \\
    & =  - \frac{1}{2}\E_{z}\left[(Uz +\mu - \mu)^T(UU^T)^{-1}(Uz + \mu-\mu)\right] \\
    & - \frac{k}{2}\log (2\pi) - \frac{1}{2}\log|\det(UU^T)| \\
    & = - \frac{1}{2}\E_{z}\left[z^TU^T(UU^T)^{-1}Uz\right] \\
    & - \frac{k}{2}\log (2\pi) - \frac{1}{2}\log|\det(UU^T)| \\
    & =  -\frac{k}{2}- \frac{k}{2}\log (2\pi) - \log\det(U)
\end{align*} where we used the second half of Lemma \ref{lem:helpful-expectations} in the last equality and $\log |\det(UU^T)| = 2\log \det(U)$ in the last equality. Adding the terms achieves the desired result.
\end{proof}

\begin{proof}[Proof of Lemma \ref{lem:helpful-expectations}] For the first claim, let $U = ESV^T$ denote the singular value decomposition of $U$ with $E$ and $V$ square and unitary. Then by the rotational invariance of the Gaussian distribution and the $\ell_2$-norm, we have \begin{align*}
    \E_z\|Uz\|_2^2 = \E_z\|ESV^Tz\|_2^2 & = \E_z\|Sz\|_2^2 \\
    & = \sum_{i=1}^{\mathrm{rank}(U)}\sigma_i(U)^2\E_{z_i}[z_i^2] \\
    & = \sum_{i=1}^{\mathrm{rank}(U)}\sigma_i(U)^2 = \|U\|_F^2.
\end{align*} For the second claim, recall that $z^TU^T(UU^T)^{-1}Uz = \langle U^T(UU^T)^{-1}U, zz^T\rangle$ where $\langle A,B\rangle = \mathrm{Tr}(A^TB)$ is the Hilbert-Schmidt inner product for matrices. Using linearity of the expectation and the cyclic property of the trace operator, we get \begin{align*}
\E_z[z^TU^T(UU^T)^{-1}Uz]& = \E_z\langle U^T(UU^T)^{-1}U, zz^T\rangle  \\
& =\langle U^T(UU^T)^{-1}U, \E_z[zz^T]\rangle \\
& =\langle U^T(UU^T)^{-1}U, I_n\rangle  \\
& = \mathrm{Tr}(U^T(UU^T)^{-1}U) \\
& = \mathrm{Tr}(UU^T(UU^T)^{-1}) \\
& = \mathrm{Tr}(I_d) = d.
\end{align*}
\end{proof}

We now prove Theorem \ref{thm:main-estimators}. Our proof essentially computes the stationary conditions of the loss $\Lcal$ and computes the first-order critical points directly. 

\begin{proof}[Proof of Theorem \ref{thm:main-estimators}]

 Using standard matrix calculus, the gradient of the objective $\Lcal$ with respect to each parameter is given by  \begin{align*}
    \nabla_{G_{\theta}}\Lcal & = \frac{1}{\sigma^2}\Big[\frac{1}{N}\sum_{i=1}^NA^TAG_{\theta}(U_iU_i^T + \mu_i\mu_i^T) \\
    & - \frac{1}{N}\sum_{i=1}^NA^Ty^{(i)}\mu_i^T\Big], \\
    \nabla_{U_i}\Lcal & =  \frac{1}{\sigma^2}G_{\theta}^TA^TAG_{\theta}U_i + U_i - (U_i^T)^{-1},\ \text{and}\\
    \nabla_{\mu_i}\Lcal & = \frac{1}{\sigma^2}(G_{\theta}^TA^TAG_{\theta}\mu_i - G_{\theta}^TA^Ty^{(i)}) + \mu_i.
\end{align*} For a fixed $G_{\theta}$, we can directly compute which $\mu_i$ and $U_i$ satisfy the stationary equations $\nabla_{\mu_i}\Lcal = 0$ and $\nabla_{U_i} \Lcal = 0$: \begin{align*}
    \mu_i & = (G_{\theta}^TA^TAG_{\theta} + \sigma^2I)^{-1}G_{\theta}^TA^Ty^{(i)} \text{and}\\
    U_i & = \sigma(G_{\theta}^TA^TAG_{\theta} + \sigma^2I)^{-1/2}\tilde{R}
\end{align*} where $\tilde{R}$ is an arbitrary orthogonal matrix.

We now calculate, given such $\mu_i$ and $U_i$, which $G_{\theta}$ satisfies the condition $\nabla_{G_{\theta}}\Lcal = 0$. For notational simplicity, let $\Sigma_{\sigma} := G_{\theta}^TA^TAG_{\theta} + \sigma^2I$. Plugging in $\mu_i$ and $U_i$ into the stationary condition $\nabla_{G_{\theta}}\Lcal = 0$, this equation reads \begin{align*}
    A^TAG_{\theta}(\sigma^2 I_k + & \Sigma_{\sigma}^{-1}G_{\theta}^TA^TY_NAG_{\theta})\Sigma_{\sigma}^{-1} \\
    & = A^TY_NAG_{\theta}\Sigma_{\sigma}^{-1}.
\end{align*} Since $A$ has rank $m$, $AA^T$ is invertible, so we can multiply both sides of the equation on the left by $(AA^T)^{-1}A$ and the right by $\Sigma_{\sigma}$ to obtain \begin{align}
    (I_m - AG_{\theta}\Sigma_{\sigma}^{-1}G_{\theta}^TA^T)Y_N AG_{\theta} = \sigma^2 AG_{\theta}. \label{eq:stat_cond}
\end{align} Now, consider the SVD of $AG_{\theta} = ESV^T$ where $E \in \R^{m \times m}$ and $V \in \R^{k \times k}$ are orthogonal and $S \in \R^{m \times k}$ is rectangular with non-zero diagonal submatrix containing the singular values of $AG_{\theta}$ denoted by $\tilde{S} \in \R^{k \times k}$. First, observe that \begin{align*}
    AG_{\theta}\Sigma_{\sigma}^{-1}G_{\theta}^TA^T & = ESV^T(VS^TSV^T + \sigma^2 I_k)^{-1}VS^TE^T \\
    & = ESV^T(V\tilde{S}^2 V^T + \sigma^2 I_k)^{-1}VS^TE^T \\
    & = ESV^TV(\tilde{S}^2  + \sigma^2 I_k)^{-1}V^TVS^TE^T \\
    & = ES(\tilde{S}^2 + \sigma^2I_k)^{-1}S^TE^T \\
    & = E\left[\begin{array}{cc}
        \tilde{S}^2(\tilde{S}^2 + \sigma^2 I_k)^{-1} & 0 \\
        0 & 0
    \end{array}\right]E^T.
\end{align*} This implies that $I_m - AG_{\theta}\Sigma_{\sigma}^{-1}G_{\theta}^TA^T = E\hat{S}_{\sigma}E^T$ where \begin{align*}
    \hat{S}_{\sigma} := \left[\begin{array}{cc}
        I_k - \tilde{S}^2(\tilde{S}^2 + \sigma^2 I_k)^{-1}  &  0\\
        0 &  I_{m-k}
    \end{array}\right].
\end{align*} Note that for $\sigma > 0$, $\hat{S}_{\sigma}$ is invertible. Hence, our stationary equation \eqref{eq:stat_cond} reads \begin{align*}
    E\hat{S}_{\sigma}E^TY_NESV^T = \sigma^2 ESV^T.
\end{align*} Using the orthogonality of $E$ and $V$ and invertibility of $\hat{S}_{\sigma}$, we can further simplify the above equation to \begin{align*}
    Y_N ES = \sigma^2 E\hat{S}_{\sigma}^{-1}S.
\end{align*} 

Now, suppose $S_{ii} = \sigma_i(AG_{\theta}) = \sigma_i \neq 0$. Then we have that the $i$-th column of $E$, denoted $e_i$, satisfies \begin{align*}
    Y_N e_i = \sigma^2\left(1 - \frac{\sigma_i^2}{\sigma_i^2 + \sigma^2}\right)^{-1}e_i.
\end{align*} That is, $e_i$ is an eigenvector of $Y_N$ with eigenvalue $\lambda_i = \sigma^2(1 - \sigma_i^2/(\sigma_i^2 + \sigma^2))^{-1}.$ Observe that this equation is satisfied if and only if $\sigma_i = \pm\sqrt{\lambda_i - \sigma^2}$. Since singular values are non-negative, we must have $\sigma_i = \sqrt{\lambda_i - \sigma^2}.$ Note that if $S_{ii} = 0$, then $e_i$ can be arbitrary. Thus, we obtain that in order for $G_{\theta}$ to satisfy the stationary condition, it must satisfy \begin{align*}
    AG_{\theta} = E_k(L_k - \sigma^2 I)^{1/2}R
\end{align*} where $E_k$ contains the top-$\min\{\mathrm{rank}(Y_N), k\}$ eigenvectors of $Y_N$, $L_k$ is a diagonal matrix whose $i$-th entry is $\lambda_i$ with corresponding eigenvector $e_i$ and $\sigma^2$ otherwise, and $R$ is an arbitrary orthogonal matrix.

\end{proof}

\subsection{Application: Denoising Data From a Subspace} \label{appx:supplementary-theory-denoising-ex}

In order to gain further intuition about Theorem \ref{thm:main-estimators}, it will be useful to consider an example where the data explicitly lies in a low-dimensional subspace. We will then compare our inferred estimator with an estimator that has access to complete prior knowledge about the data.

Consider the denoising problem $A = I$ and suppose that each ground-truth image is drawn at random from a $k$-dimensional subspace as follows: $x^{(i)} = Gz^{(i)}$ where $G \in \R^{n \times k}$ and $z^{(i)} \sim \mathcal{N}(0,I)$. If we had access to the ground-truth generating matrix $G$ and knew the underlying distribution in the latent space $\mathcal{N}(0,I)$, then one could denoise the $i$-th image $y^{(i)}$ by returning the following MAP estimate $\hat{x}^{(i)}$ in the range of $G$: \begin{align*}
   \hat{x}^{(i)} & := G\hat{z}^{(i)},\\
   \text{where}\ \hat{z}^{(i)} & := \argmin_{z \in \R^k} \frac{1}{2\sigma^2}\|y^{(i)} - Gz\|_2^2 + \frac{1}{2}\|z\|_2^2 \\
   & = (G^TG + \sigma^2I)^{-1}G^Ty^{(i)}.
\end{align*} Note that this is the mean of the posterior $p(x|y^{(i)},G)$. Now, to compare solutions, observe that our (mean) estimator given by Theorem \ref{thm:main-estimators} is  \begin{align*}
    G_*\mu_i^* = G_*(G_*^TG_* + \sigma^2I)^{-1}G_*^Ty^{(i)}.
\end{align*} Recall that $G_*$ contains the top eigenvectors and eigenvalues of $Y_N$, the sample covariance of the measurements. In the case that our data $x^{(i)} = Gz^{(i)}$, note that for sufficiently many examples $N$, the sample covariance $Y_N$ is close to its mean $\E_{z,\eta} Y_N$. That is, for large $N$, $$Y_N \approx \E_{z,\eta} Y_N =  GG^T + \sigma^2 I.$$ To see this more concretely, let $G = ELV^T$ denote the SVD of $G$. Then $GG^T + \sigma^2 I = E\hat{L}_{\sigma}E^T$ where $\hat{L}_{\sigma}$ has $\sigma_i(G)^2 + \sigma^2$ in the first $k$ diagonal entries and $\sigma^2$ in the remaining $n-k$ entries. So in the limit of increasing data, $G_*$ estimates the top-$k$ eigenvectors of $\E_{z,\eta} Y_N$, which are precisely the top-$k$ principal components of $G$. Moreover, the singular values of $G_*$ are $\sqrt{\lambda_i - \sigma^2}$ where $\lambda_i = \lambda_i(Y_N)$. As data increases, $\lambda_i(Y_N) \approx \lambda_i(\E_{z,\eta} Y_N) =  \sigma_i(G)^2 + \sigma^2$ so that $\sigma_i(G_*) = \sqrt{\lambda_i - \sigma^2} \approx \sigma_i(G)$. Thus the singular values of $G_*$ also approach those of $G$. 

A natural question to consider is how large the number of examples $N$ would need to be in order for such an approximation to hold. Ideally, this approximation would hold as long as the number of examples $N$ scales like the intrinsic dimensionality of the underlying images.  For example, the following result shows that in the noiseless case, this approximation holds as long as the number of examples $N$ scales like $k$, the dimension of the subspace. The proof follows from standard matrix concentration bounds, e.g., Theorem 4.6.1 in \citePhys{VershyninHDPSupp}.
\begin{prop}
Fix $G \in \R^{n \times k}$ and let $y^{(i)} = Gz^{(i)}$ where $z^{(i)} \sim \mathcal{N}(0,I)$ for $i \in [N]$. Then, for any $\varepsilon > 0$, there exists positive absolute constants $C_1$ and $C_2$ such that the following holds: if $N \geqslant C_1\varepsilon^{-2} k$, then with probability $1 - 2e^{- k}$ the sample covariance $Y_N := \frac{1}{N}\sum_{i=1}^N y^{(i)}(y^{(i)})^T$ obeys \begin{align*}
    \left\|Y_N - GG^T\right\| \leqslant C_2 \|G\|^2 \varepsilon.
\end{align*}
\end{prop}

Hence, in the case of data lying explicitly in a low-dimensional subspace, our estimator returns an approximation of the solution found via MAP estimation, which can only be computed with complete prior knowledge of the underlying image structure. In contrast, our estimator infers this subspace from measurement information only and its performance improves as the number of examples increases. The approximation quality of the estimator is intimately related to the number of examples and the intrinsic dimensionality of the images, which we show can be quantified in simple settings. While our main result focused on IGMs inferred with linear models, it is an interesting future direction to theoretically understand properties of non-linear IGMs given by deep networks and what can be said for more complicated forward models. We leave this for future work.

\subsection{Additional Theoretical Results}

We can further analyze properties of the model found by minimizing $\Lcal$ from Theorem \ref{thm:main-estimators}. In particular, we can characterize specific properties of the inferred mean $\mu_i^*$ as a function of the amount of noise in our measurements and a fixed IGM $G_{\theta} \in \R^{n \times k}$. Specifically, given measurements $y^{(i)}$ and an IGM $G_{\theta}$, recall that the inferred mean computed in Theorem \ref{thm:main-estimators} is given by \begin{align*}
    \mu_i^*(\sigma) := \sigma^{-2}(\sigma^{-2}\cdot G_{\theta}^TA^TAG_{\theta} + I)^{-1}G_{\theta}^TA^Ty^{(i)}
\end{align*} We will give an explicit form for $\mu_i^*(\sigma)$ as the amount of noise goes to $0$ and show that, in the case our IGM is correct ($G_{\theta} = G$ for some ground-truth $G$), the mean recovers the ideal latent code with vanishing noise.

\begin{prop}
Suppose $A \in \R^{m \times n}$ and $G \in \R^{n \times k}$ with $k \leqslant m \leqslant n$. Let $y^{(i)} := AGz_i + \eta_i$ where $z_i \in \R^k$ and $\eta_i \sim \mathcal{N}(0,\sigma^2 I)$ for $i = 1,\dots,N$. Then, with probability $1$, for any $G_{\theta} \in \R^{n \times k}$, the estimator $\mu_i^*(\sigma)$ for $z_i$ defined above satisfies \begin{align*}
    \lim_{\sigma \rightarrow 0^+} \mu_i^*(\sigma) = (AG_{\theta})^{\dagger}AGz_i\ \forall\ i \in [N].
\end{align*} In the special case that $G_{\theta} = G$ and $\mathrm{rank}(AG) = k$, then \begin{align*}
    \lim_{\sigma \rightarrow 0^+} \mu_i^*(\sigma) = z_i\ \forall\ i \in [N].
\end{align*} Moreover, if $\mathrm{rank}(AG_{\theta}) = \mathrm{rank}(AG) = k$, then we get that the following error bound holds for all $i \in [N]$ with probability at least $1 - Ne^{-\gamma m}$ for some positive absolute constant $\gamma$: 

\begin{align*}
    \|\mu_i^*(\sigma) - z_i\|_2 & \leqslant \|AGz_i\|_2\sigma_{k}(AG_{\theta})^{-3}\cdot\sigma^2 \\
    & + 2\sigma_{k}(AG_{\theta})^{-1}\sqrt{m} \cdot \sigma  \\
    & + \sqrt{2}\|(AG_{\theta})^{\dagger}\|\|(AG)^{\dagger}\|\cdot\|A(G_{\theta} - G)\|.
\end{align*} 
\end{prop}
\begin{proof}
 For notational simplicity, set $\Sigma_{\theta,\sigma} := \sigma^{-2} \cdot G_{\theta}^TA^TAG_{\theta} + I$. Then \begin{align*}
    \mu_i^*(\sigma) & = \sigma^{-2}\Sigma_{\theta,\sigma}^{-1}G_{\theta}^TA^Ty^{(i)} \\
    & = \sigma^{-2}\Sigma_{\theta,\sigma}^{-1}G_{\theta}^TA^T(AGz_i + \eta_i) \\
    & = \sigma^{-2}\Sigma_{\theta,\sigma}^{-1}G_{\theta}^TA^TAGz_i + \sigma^{-2}\Sigma_{\theta,\sigma}^{-1}G_{\theta}^TA^T\eta_i \\
    & =\sigma^{-2} \Sigma_{\theta,\sigma}^{-1}G_{\theta}^TA^TAGz_i + \sigma^{-1} \cdot \Sigma_{\theta,\sigma}^{-1}G_{\theta}^TA^T\tilde{\eta}_i
    \end{align*} where we note that any $\eta_i \sim \mathcal{N}(0,\sigma^2I)$ can be written as $\eta_i = \sigma\cdot\tilde{\eta}_i$ where $\tilde{\eta}_i \sim \mathcal{N}(0,I)$. Observe that \begin{align*}
        \Sigma_{\theta,\sigma}^{-1} & = (\sigma^{-2}G_{\theta}^TA^TAG_{\theta} + \sigma^{-2} \cdot\sigma^{2}I)^{-1} \\
        & = \sigma^2(G_{\theta}^TA^TAG_{\theta} + \sigma^2 I)^{-1}
    \end{align*} since $(cA)^{-1} = c^{-1}A^{-1}$ for $c \neq 0$. This simplifies $\mu_i^*(\sigma)$ to \begin{align*}
        \mu_i^*(\sigma) & = (G_{\theta}^TA^TAG_{\theta} + \sigma^2I)^{-1}G_{\theta}^TA^TAGz_i \\
        & + \sigma \cdot (G_{\theta}^TA^TAG_{\theta} + \sigma^2I)^{-1}G_{\theta}^TA^T\tilde{\eta}_i.
    \end{align*} To compute the limit of the expression above, we recall that the pseudoinverse of a matrix $M$ is a limit of the form \begin{align*}
        M^{\dagger} & := \lim_{\delta \rightarrow 0^+} (M^TM + \delta I)^{-1}M^T \\
        & = \lim_{\delta \rightarrow 0^+} M^T(MM^T + \delta I)^{-1}.
    \end{align*} Note that this quantity is defined even if the inverses of $MM^T$ or $M^TM$ exist. Applying this to $M = AG_{\theta}$, we observe that \begin{align*}
       \lim_{\sigma \rightarrow 0^+} (G_{\theta}^TA^TAG_{\theta} + \sigma^2I)^{-1}G_{\theta}^TA^TAGz_i & = (AG_{\theta})^{\dagger}AGz_i
    \end{align*} and \begin{align*}
        & \lim_{\sigma \rightarrow 0^+}  \sigma (G_{\theta}^TA^TAG_{\theta} + \sigma^2I)^{-1}G_{\theta}^TA^T\tilde{\eta}_i \\
        & = \left(\lim_{\sigma \rightarrow 0^+} \sigma \right)\left(\lim_{\sigma \rightarrow 0^+} (G_{\theta}^TA^TAG_{\theta} + \sigma^2I)^{-1}G_{\theta}^TA^T\tilde{\eta}_i\right) \\
        & = 0 \cdot (AG_{\theta})^{\dagger} = 0.
    \end{align*}  Combining the above two limits gives $\lim_{\sigma \rightarrow 0^+} \mu_i^*(\sigma) = (AG_{\theta})^{\dagger} AGz_i$ as claimed.
    
    For the special case, since $AG \in \R^{m \times k}$ has $\mathrm{rank}(AG) = k$, the pseudoinverse of $AG$ is given by $(AG)^{\dagger} = (G^TA^TAG)^{-1}G^TA^T$. Applying this to $M = AG$, we observe that \begin{align*}
       \lim_{\sigma \rightarrow 0^+} & (G^TA^TAG + \sigma^2I)^{-1}G^TA^TAGz_i \\
       & = (AG)^{\dagger}AGz_i  = z_i.
    \end{align*}
    
    Now, for the error bound, we note that for $\sigma \geqslant 0$, we have by the triangle inequality, \begin{align*}
        \|\mu_i^*(\sigma) - z_i\|_2 & \leqslant \|\mu_i^{*}(\sigma) - \lim_{\sigma \rightarrow 0^+}\mu_i^*(\sigma)\|_2 \\
        & + \|\lim_{\sigma \rightarrow 0^+}\mu_i^*(\sigma) - z_i\|_2
    \end{align*} so we bound each term separately. Recall that we have shown $\lim_{\sigma \rightarrow 0^+}\mu_i^*(\sigma) = (AG_{\theta})^{\dagger}AGz_i$. For the first term, observe that \begin{align*}
       & \|\mu_i^*(\sigma) - (AG_{\theta})^{\dagger}AGz_i\|_2  \\
        & \leqslant  \|(G_{\theta}^TA^TAG_{\theta} + \sigma^2I)^{-1}G_{\theta}^TA^TAGz_i - (AG_{\theta})^{\dagger}AGz_i\|  \\
        &+ \sigma\cdot\|(G_{\theta}^TA^TAG_{\theta} + \sigma^2)^{-1}G_{\theta}^TA^T\tilde{\eta}_i\|_2 \\
        & \leqslant \|(G_{\theta}^TA^TAG_{\theta} + \sigma^2I)^{-1}G_{\theta}^TA^T - (AG_{\theta})^{\dagger}\|\|AGz_i\|_2  \\
        & + \sigma\cdot\|(G_{\theta}^TA^TAG_{\theta} + \sigma^2)^{-1}G_{\theta}^TA^T\tilde{\eta}_i\|_2
    \end{align*} Let $A G_{\theta} = USV^T$ denote the singular value decomposition of $AG_{\theta}$. Then $G_{\theta}^TA^TAG_{\theta} = VSU^TUSV^T = VS^2V^T$. This implies that \begin{align*}
        (G_{\theta}^TA^TAG_{\theta} & + \sigma^2I)^{-1}G_{\theta}A^T \\
        & = (VS^2V^T + \sigma^2 I)^{-1}VSU^T.
    \end{align*} Moreover, since $V$ is orthogonal, we have that \begin{align*}
        (VS^2V^T + \sigma^2I)^{-1}&  = (VS^2V^T + \sigma^2VV^T)^{-1} \\
        & = V(S^2 + \sigma^2 I)^{-1}V^T
    \end{align*}Thus, for $\sigma, \alpha \geqslant 0$, we have that \begin{align*}
        & [(G_{\theta}^TA^TAG_{\theta} + \sigma^2I)^{-1}   - (G_{\theta}^TA^TAG_{\theta} + \alpha^2I)^{-1}]G_{\theta}A^T  \\
        & = \left[(VS^2V^T + \sigma^2I)^{-1} - (VS^2V^T + \alpha^2 I)^{-1}\right]VSU^T \\
        & = V\left[(S^2 + \sigma^2I)^{-1} - (S^2 + \alpha^2 I)^{-1}\right]V^TVSU^T \\
        & = V\left[(S^2 + \sigma^2I)^{-1} - (S^2 + \alpha^2 I)^{-1}\right]SU^T.
    \end{align*} Thus, we conclude by rotational invariance of the spectral norm, we obtain \begin{align*}
         &\|(G_{\theta}^TA^TAG_{\theta} + \sigma^2I)^{-1}G_{\theta}A^T \\
         & - (G_{\theta}^TA^TAG_{\theta} + \alpha^2I)^{-1}G_{\theta}A^T\| \\
        & = \left\|\left[(S^2 + \sigma^2I)^{-1} - (S^2 + \alpha^2 I)^{-1}\right]S\right\|.
    \end{align*} Recall that $S$ is a diagonal matrix containing the singular values of $AG_{\theta}$. Let $\sigma_i$ denote the $i$-th singular value of $AG_{\theta}$ for $i = 1,\dots,k$. Then by direct calculation, we have that \begin{align*}
       &  \left\|\left[(S^2 + \sigma^2I)^{-1} - (S^2 + \alpha^2 I)^{-1}\right]S\right\| \\
       & =\max_{i \in [k]}\left|\sigma_i\cdot\left(\frac{1}{\sigma_i^2 + \sigma^2} - \frac{1}{\sigma_i^2 + \alpha^2}\right)\right| \\
        & = \max_{i \in [k]}\left|\frac{\sigma_i(\sigma^2 - \alpha^2)}{(\sigma_i^2 + \sigma^2)(\sigma_i^2 + \alpha^2)}\right|.
    \end{align*} Taking the limit as $\alpha \rightarrow 0^+$, we obtain \begin{align*}
         & \|(G_{\theta}^TA^TAG_{\theta} + \sigma^2I)^{-1}G_{\theta}^TA^T - (AG_{\theta})^{\dagger}\| \\
         & =  \max_{i \in [k]}\frac{\sigma^2}{\sigma_i(\sigma_i^2 + \sigma^2)} \\
         & \leqslant \sigma_k^{-3}\cdot \sigma^2.
    \end{align*} This additionally shows that for any $\sigma \geqslant 0$, \begin{align*}
        \|(G_{\theta}^TA^TAG_{\theta} & + \sigma^2I)^{-1}G_{\theta}A^T\| \\
        & = \left\|(S^2+\sigma^2I)^{-1}S\right\|\\
        & \leqslant \sigma_k^{-1}.
    \end{align*} Let $E_N$ be the event that $\max_{i \in [N]}\|\tilde{\eta}_i\|_2 \leqslant 2\sqrt{m}$. By a single-tailed variant of Corollary 5.17 in \citePhys{VershyninHDPSupp}, we have $\mathbb{P}(E_N) \geqslant 1 - Ne^{-\gamma m}$ for some absolute constant $\gamma > 0$. Thus, with the same probability, we have that the following bound holds: \begin{align*}
        \sigma\cdot\|(G_{\theta}^TA^TAG_{\theta} + \sigma^2)^{-1}G_{\theta}^TA^T\tilde{\eta}_i\|_2 \leqslant 2\sigma \sigma_k^{-1}\sqrt{m}.
    \end{align*}
    
    For the term $\|\lim_{\sigma \rightarrow 0^+}\mu_i^*(\sigma) - z_i\|_2$, note that we have \begin{align*}
        \|(AG_{\theta})^{\dagger} AGz_i - z_i\|_2 & = \|(AG_{\theta})^{\dagger} AGz_i - (AG)^{\dagger}AGz_i\|_2 \\
        & \leqslant \|(AG_{\theta})^{\dagger} - (AG)^{\dagger}\|\cdot\|AGz_i\|_2
    \end{align*} where $\|\cdot\|$ is the spectral norm. Assuming $\mathrm{rank}(AG_{\theta}) = \mathrm{rank}(AG) = k$, we can apply Theorem 4.1 in \citePhys{Wedin1973} to bound \begin{align*}
        \|(AG_{\theta})^{\dagger} & - (AG)^{\dagger}\| \\
        & \leqslant \sqrt{2}\|(AG_{\theta})^{\dagger}\|\|(AG)^{\dagger}\|\|A(G_{\theta}-G)\|.
    \end{align*} Combining this bound with the previous equation yields the desired inequality.
\end{proof}

\subsection{Further Discussion of $\ELBOProxy$}

In Section \ref{sec:ELBO-intro}, we showed that the $\ELBO$ is equal to the $\ELBOProxy$ for injective or invertible generators $G$ with approximate posteriors $h_{\phi}$ induced by $G$. While this equality may not hold for non-injective networks $G$, we expect this proxy to be a reasonable estimate for the ELBO via the following heuristic argument. Consider the distribution $G\sharp p(z|G)$ convolved with isotropic Gaussian noise: $p_{\sigma}(x|G) = \int p_{\sigma}(x|z,G)p(z|G)dz$ where $p_{\sigma}(x|z,G) = \mathcal{N}(G(z), \sigma^2 I)$ for $\sigma \ll 1$. If $x \sim h_{\phi}= G\sharp q_{\phi}$, then $x = G(z_0)$ for some $z_0 \sim q_{\phi}$. For any small $\delta > 0$, observe that the likelihood for such an $x$ can decomposed into two terms: \begin{align*}
    p_{\sigma}(x|G) & = \int p_{\sigma}(G(z_0)| z,G) p(z|G)dz \\
    & = \int_{\|z-z_0\|_2\leqslant \delta} p_{\sigma}(G(z_0)| z,G) p(z|G)dz  \\
    & +\int_{\|z-z_0\|_2> \delta} p_{\sigma}(G(z_0)| z,G) p(z|G)dz.
\end{align*} For smooth generators $G$ and $\delta$ sufficiently small, we would expect $G(z) \approx G(z_0)$ for $z \approx z_0$ so that the density $p_{\sigma}(G(z_0)|z,G))p(z|G)$ is constant in a neighborhood of $z_0$, making the first term $\approx C_{\sigma} p(z_0|G)$ for some constant $C_{\sigma}$. Moreover, if $G(z_0)$ is of sufficiently high likelihood, we would expect the second term to be negligible for $\sigma \ll 1$ so that $\E_{x\sim h_{\phi}}[\log p_{\sigma}(x|G)] \approx \E_{z \sim q_{\phi}}[\log p(z|G)]$ (up to a constant). 

\section{Full Model Selection Results}

In this section, we expand upon the model selection results using the negative $\ELBOProxy$ in Section \ref{sec:ELBO-intro} in the main text. The experimental setup is the same, where we aim to select a generative model for a class of MNIST in two different inverse problems: denoising and phase retrieval. The full array of negative $\ELBOProxy$ values for denoising is given in Fig. \ref{fig:denoising-model-selection} and for phase retrieval in Fig. \ref{fig:PR-model-selection}. 

\begin{figure}[ht]
    \centering
    \includegraphics[width=.49\textwidth]{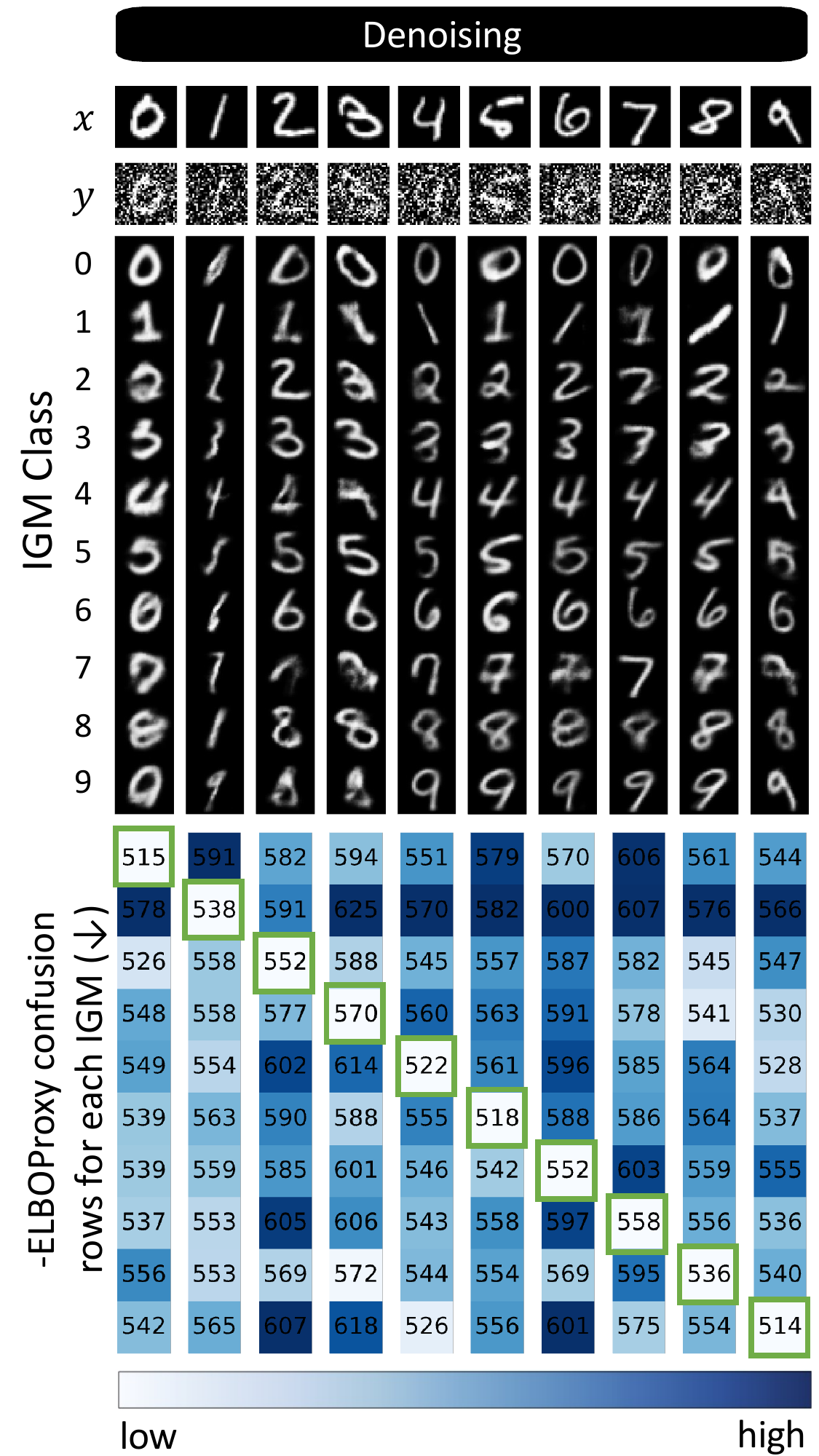}
    \caption{\textbf{Full model selection results for denoising.} We showcase further results on model selection with generative networks for denoising. Top: the two topmost rows correspond to the ground-truth image $x_c$ and the noisy measurements $y_c$. Middle: we show the means of the variational distributions given by the IGM trained on a particular class. 
    Bottom: each column of the array corresponds to the negative $\ELBOProxy$ achieved by each model in reconstructing the images. Here, lower is better. We highlight entries with green boxes with the best negative $\ELBOProxy$ values in each row.}
    \label{fig:denoising-model-selection}
\end{figure}

\begin{figure}[ht]
    \centering
    \includegraphics[width=.49\textwidth]{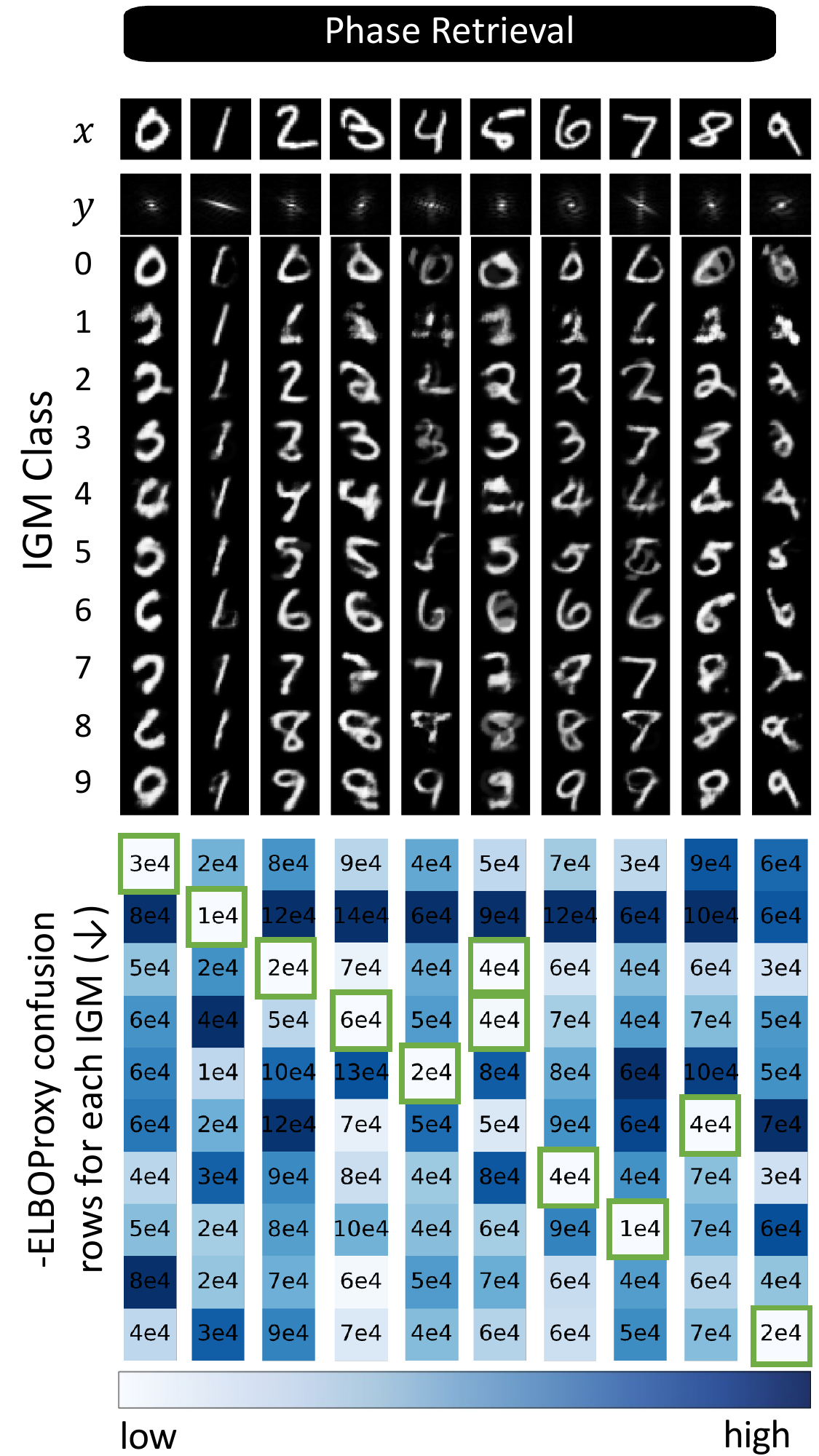}
    \caption{\textbf{Full model selection results for phase retrieval. }We showcase further results on model selection with generative networks for phase retrieval. Top: the two topmost rows correspond to the ground-truth image $x_c$ and the log of the noisy measurements $y_c$. Middle: we show the means of the variational distributions given by the IGM trained on a particular class. 
    Bottom: each column of the array corresponds to the negative $\ELBOProxy$ achieved by each model in reconstructing the images. Here, lower is better. We highlight entries with green boxes with the best negative $\ELBOProxy$ values in each row.}
    \label{fig:PR-model-selection}
\end{figure}

\section{Additional  results}
In this section, we expand upon the results for inferring the IGM using the methods described in Section \ref{sec:learning} of the main text.

\begin{figure*}[ht]
    \centering
    \includegraphics[width=.95\textwidth]{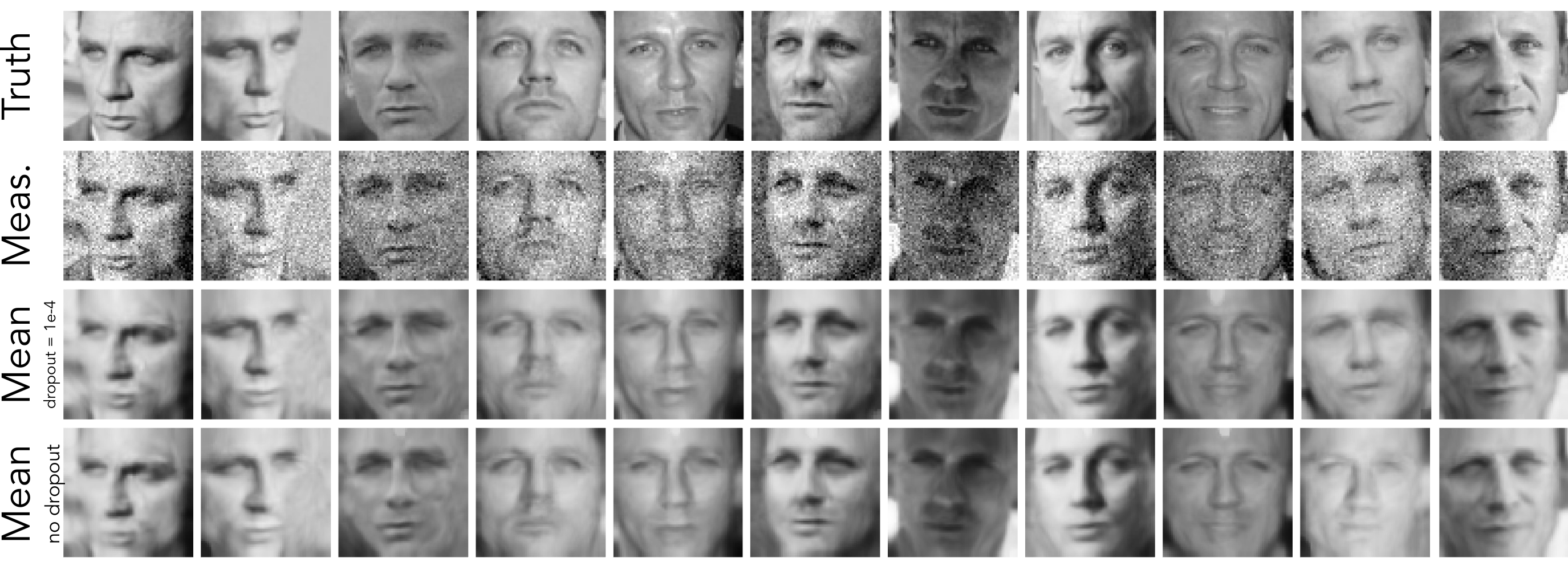}
    \caption{\textbf{Dropout ablation with denoising 95 images of celebrity A.} We demonstrate our method described in Section \ref{sec:learning} of the main body using 95 noisy images of a celebrity. Here we show the ground-truth (row 1), noisy measurements (row 2), mean reconstruction with dropout (row 3), and mean reconstruction without dropout (row 4) for a subset of the 95 different noisy images. The reconstructions with and without dropout have similar image quality. The mean PSNR for with dropout is 27.0 and for without dropout is 27.2.}
    \label{fig:bond_dropout_ablations}
\end{figure*}

\subsection{Ablations}

In Fig.~\ref{fig:bond_dropout_ablations}, we show an ablation test to validate the usage of dropout. Although the PSNR is higher when we do not have dropout, there are noticeable artifacts from the generative architecture in the reconstructions. 

\begin{figure*}[ht]
    \centering
    \includegraphics[width=.95\textwidth]{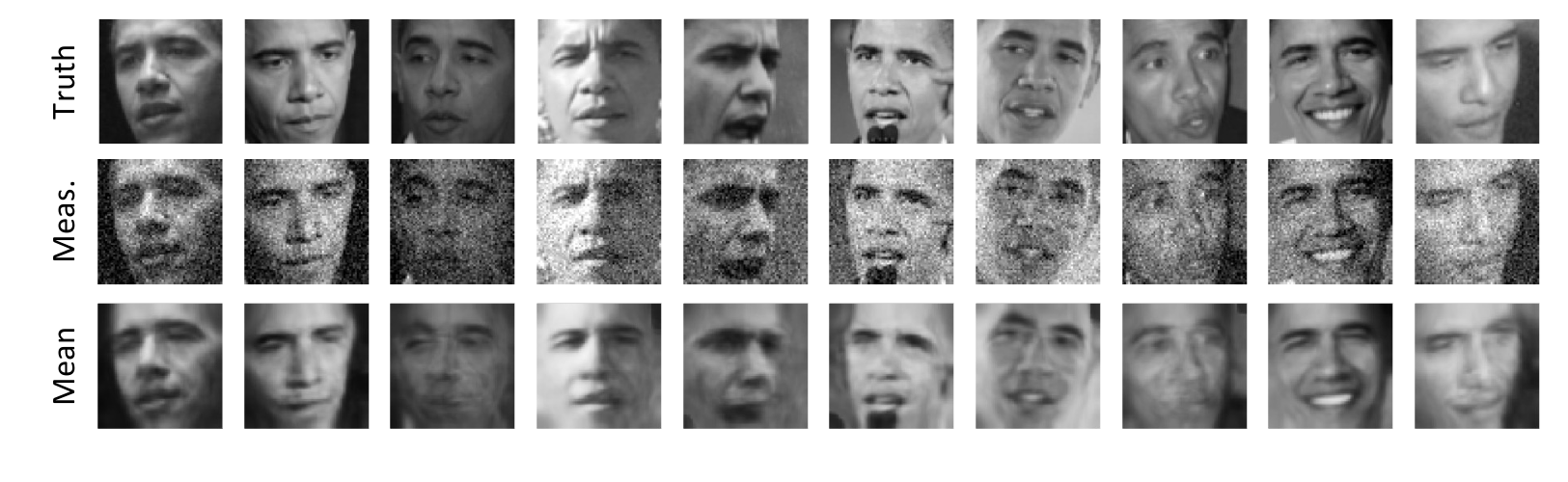}
    \caption{\textbf{Denoising 95 images of celebrity B.}  We demonstrate our method described in Section \ref{sec:learning} of the main body using 95 noisy images of a celebrity. These results are of celebrity B, but correspond to Fig. \ref{fig:Bond_Denoising} in the main body. Here we show the ground-truth (top), noisy measurements (middle), and mean reconstruction (bottom) for a subset of the 95 different noisy images. Our reconstructions are much less noisy, recovering sharper features that are hard to discern in the noisy images. Note that no predefined prior/regularizer was used in denoising.}
    \label{fig:obama_denoising}
\end{figure*}

\begin{figure*}[ht]
    \centering
    \includegraphics[width=.95\textwidth]{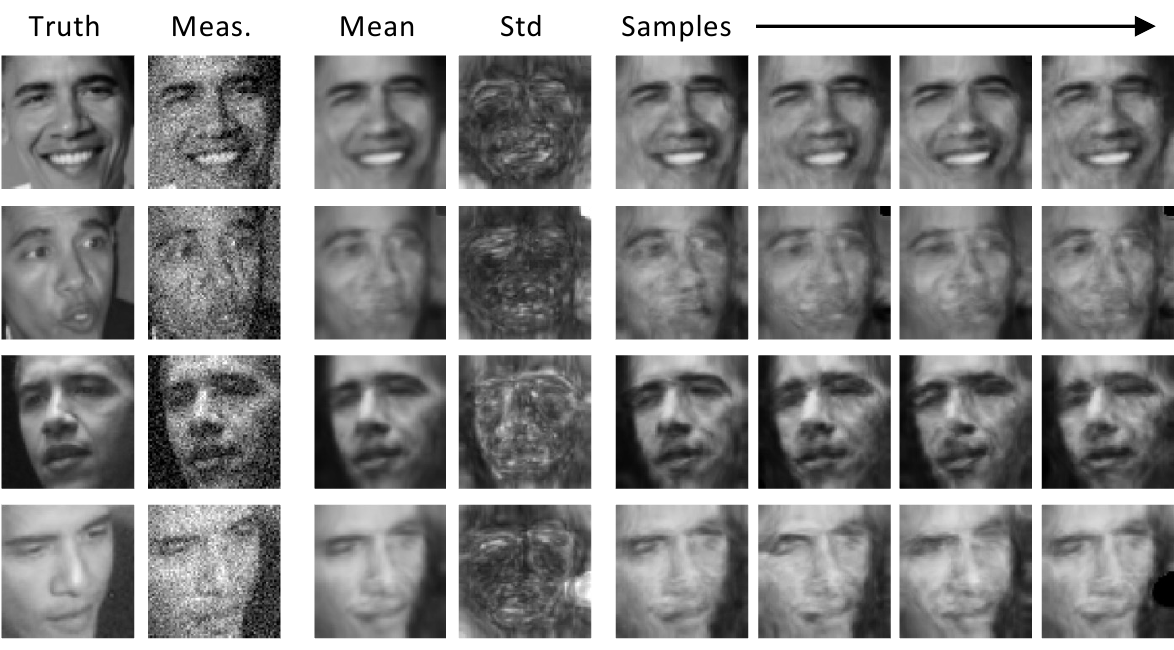}
    \caption{\textbf{Posterior samples from denoising 95 images of a celebrity.} We demonstrate our method described in Section \ref{sec:learning} of the main body to denoise 95 images. Here we show the ground-truth image, the noisy measurements, mean of the posterior, standard deviation of the posterior, and samples from the posterior.}
    \label{fig:obama_denoising_samples}
\end{figure*}

\begin{figure*}[ht]
    \centering
    \includegraphics[width=.95\textwidth]{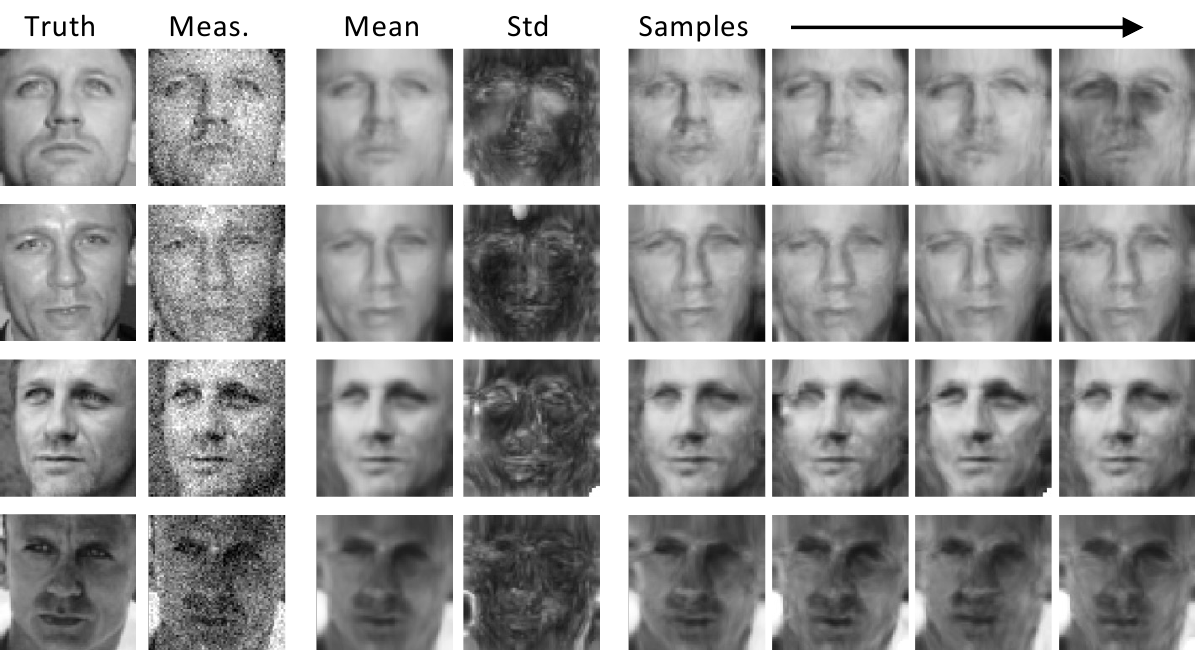}
    \caption{\textbf{Posterior samples from denoising 95 images of celebrity A.} We show additional results for denoising 95 images of celebrity A, which corresponds to the same results shown in Fig. \ref{fig:Bond_Denoising} of the main body. Here we show the ground-truth image, the noisy measurements, mean of the posterior, standard deviation of the posterior, and samples from the posterior.}
    \label{fig:bond_samples}
\end{figure*}

\subsection{Denoising}
In Fig. \ref{fig:Bond_Denoising} of the main body, we showed results for denoising 95 images of celebrity A. We include additional visualizations of the posterior by showing samples from multiple posteriors in Fig.~\ref{fig:bond_samples}, which are generated from the IGM trained as shown in Fig. \ref{fig:Bond_Denoising} in the main body. We show additional results on denoising 95 images of celebrity face B in Fig.~\ref{fig:obama_denoising} along with samples from the posterior in Fig.~\ref{fig:obama_denoising_samples}. In both cases, our reconstructions are much less noisy than the measurements and sharpen the primary facial features. 



\subsection{Phase Retrieval} \label{appx:pr-details}
We show additional results from phase retrieval with 150 images of MNIST 8's in Fig.~\ref{fig:fourier_pr_more} and Fig.~\ref{fig:gauss_pr_more}. The examples shown in Fig.~\ref{fig:fourier_pr_more} are from the Fourier phase retrieval measurements whereas the ones in Fig.~\ref{fig:gauss_pr_more} are from Gaussian phase retrieval measurements. These are in addition to those shown in Fig.~\ref{fig:Mnist_phase_retrieval} in the main text. Note that the simple, unimodal Gaussian variational distribution is not expressive enough to capture the multimodal structure of the true posterior in the Fourier phase retrieval problem.

For Table \ref{table:phase-retrieval-psnr} in the main body, we calculated the PSNRs for each measurement model in the following way: for each underlying image, we generated $1000$ samples from the generator and latent variational distribution and used cross-correlation in the frequency domain to find the most plausible shift from our reconstruction to the underlying image. After shifting the image, we calculated the PSNR to the underlying image. We then took the best PSNR amongst all samples and took the average across all $N$ image examples.


\begin{figure*}
    \centering
    \includegraphics[width=.95\textwidth]{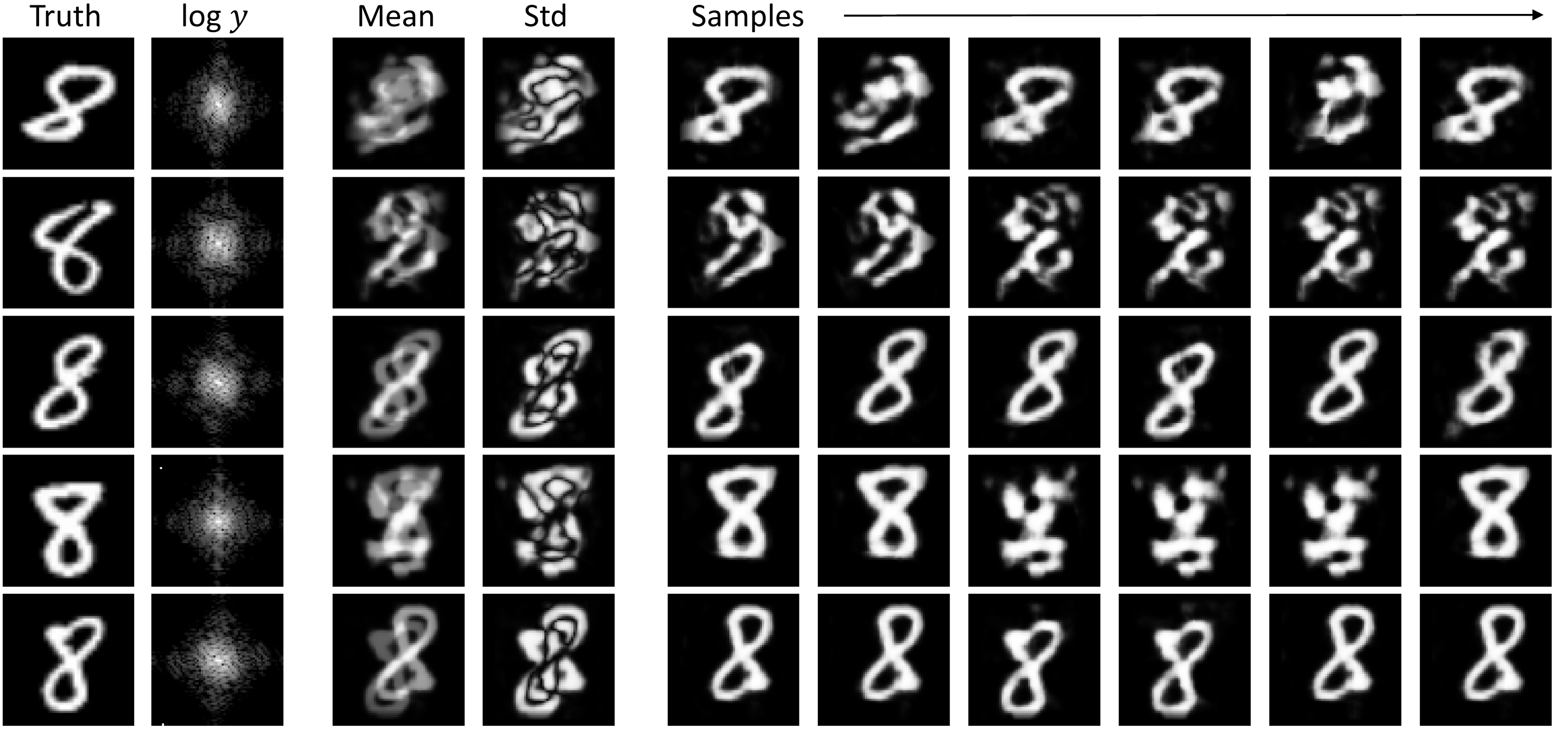}
    \caption{\textbf{Random examples from Fourier phase retrieval on MNIST 8's.} We demonstrate our method described in Section \ref{sec:learning} of the main body to perform phase retrieval on 75 images. Here we show the ground-truth image, the log magnitude of the Fourier transform, mean of the posterior, standard deviation of the posterior, and samples from the posterior. Note that the simple Gaussian latent posterior distribution is unable to capture the full multimodal structure of the true posterior. }
    \label{fig:fourier_pr_more}
\end{figure*}

\begin{figure*}
    \centering
    \includegraphics[width=.95\textwidth]{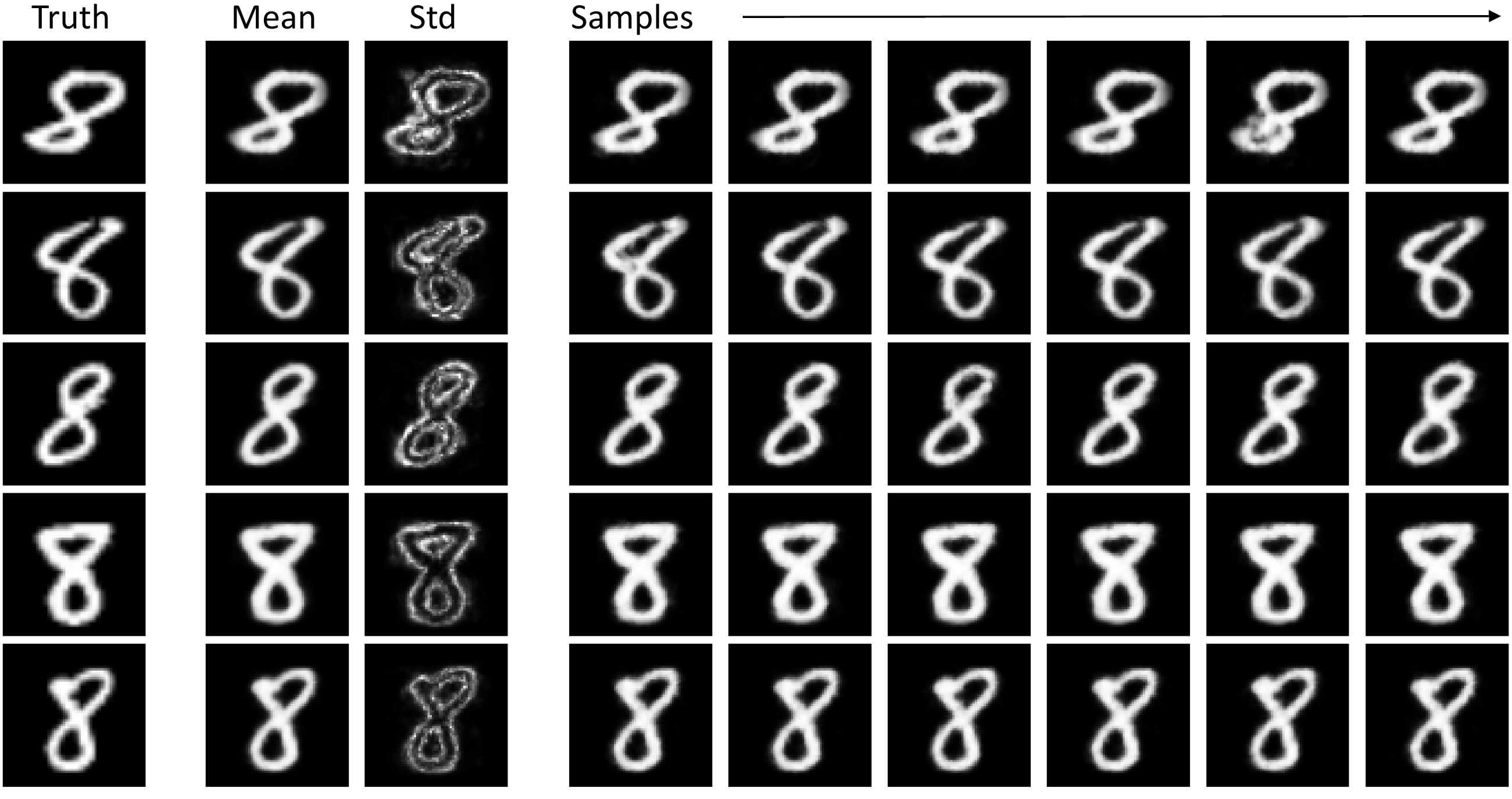}
    \caption{\textbf{Random examples from Gaussian phase retrieval on MNIST 8's.} We demonstrate our method described in Section \ref{sec:learning} of the main body to perform phase retrieval on 150 images. Here we show the ground-truth image, mean of the posterior, standard deviation of the posterior, and samples from the posterior.}
    \label{fig:gauss_pr_more}
\end{figure*}

\begin{table*}[h!]
\centering
\begin{adjustbox}{width=1\textwidth,center=\textwidth}
\begin{tabular}{l|rrrrrrrr}
\hline
& \multicolumn{8}{c}{Blur size}\\
{} &   0 $\mu$as  &   5 $\mu$as  &   10 $\mu$as &   15 $\mu$as &   20 $\mu$as &   25 $\mu$as&   30 $\mu$as &   35 $\mu$as\\
\hline
PSNR ($\uparrow$) & 26.8 & 30.6 & \textbf{33.0} & 32.4 & 30.7 & 29.3 & 28.1 & 27.3 \\
NCC ($\uparrow$) &  0.922 &  0.971 &  \textbf{0.984} &  0.978 &  0.961 &  0.941 &  0.920 &  0.901 \\
\hline
\end{tabular}
\end{adjustbox}
\caption{\textbf{Illustrating the intrinsic resolution of our result.} We show quantitative comparisons between our reconstructions and the clean underlying image blurred by varying degrees, where 0 $\mu$as is no blur and 35 $\mu$as is the highest blur. 25 $\mu$as represents the intrinsic resolution of the telescope. We find we are able to super-resolve the image by over 2x. The highest PSNR is highlighted in bold.}
\label{table:superres_table}
\end{table*}

\newpage

\subsection{Multiple Forward Models}
\begin{figure*}
    \centering
    \includegraphics[width=.95\textwidth]{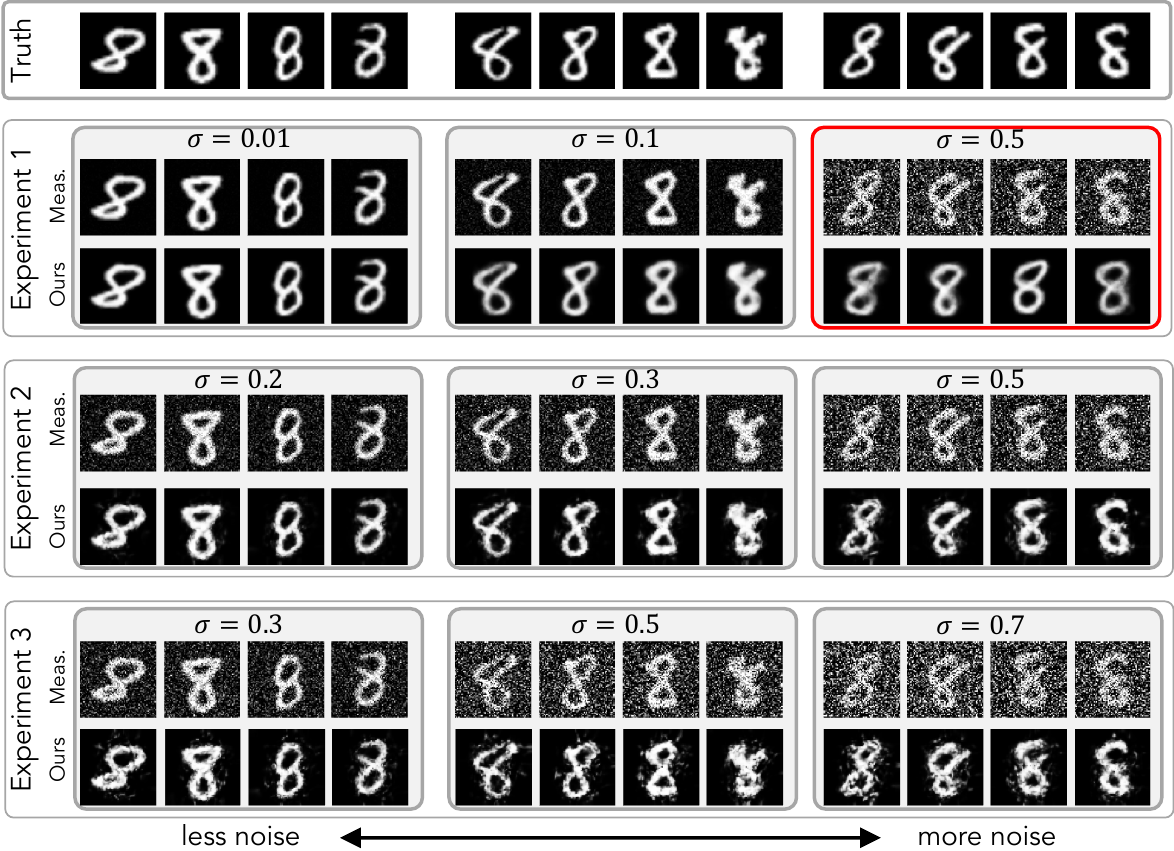}
    \caption{\textbf{Denoising many noise levels.} We demonstrate our method to perform denoising on measurements with multiple noise levels. For each experiment, we use 75 measurement examples, which are defined by  $y^{(i)} = x^{(i)} + \eta$ where $\eta \in \{\eta_1, \eta_2, \eta_3\}$ and $\eta_i \sim \mathcal{N}(0, \sigma_i^2 I)$. In Experiment 1, we use noisy measurement examples that have additive noise with standard deviations of 0.01, 0.1, and 0.5. In Experiment 2, we use noisy measurement examples that have additive noise with standard deviations of 0.2, 0.3, and 0.5. In Experiment 3, we use noisy measurement examples that have additive noise with standard deviations of 0.3, 0.5, and 0.7. We visualize the true underlying images, the measurement used for each experiment, and the mean of the image reconstruction posterior. Most of the reconstructions recover the primary features of the underlying image. However, in Experiment 1, the reconstructions of the low SNR measurements exhibit bias and do not match the true underlying images. Due to the wide range of SNRs with nearly noiseless measurements, our reconstructions overfit to the nearly clean images. The biased results are highlighted by the red box.} 
    \label{fig:multi_denoise_all}
\end{figure*}


\section{IGM as a Generator}
While the goal of our method is to improve performance in solving the underlying inverse problem, we can inspect what the inferred IGM has learned by generating samples. We show samples from different inferred IGMs in Fig.~\ref{fig:generation} for a variety of datasets and inverse problems. From left to right: IGM from denoising MNIST 8's (Fig. \ref{fig:MNIST_denoising} from the main body), IGM inferred from phase retrieval measurements of MNIST 8's, IGM from denoising celebrity face A (Fig. \ref{fig:Bond_Denoising} from the main body), IGM from denoising celebrity face B (Fig.~\ref{fig:obama_denoising}), and IGM from video reconstruction of a black hole from the black hole compressed sensing problem (Fig. \ref{fig:m87_all} from the main body).

\begin{figure*}
    \centering
    \includegraphics[width=.95\textwidth]{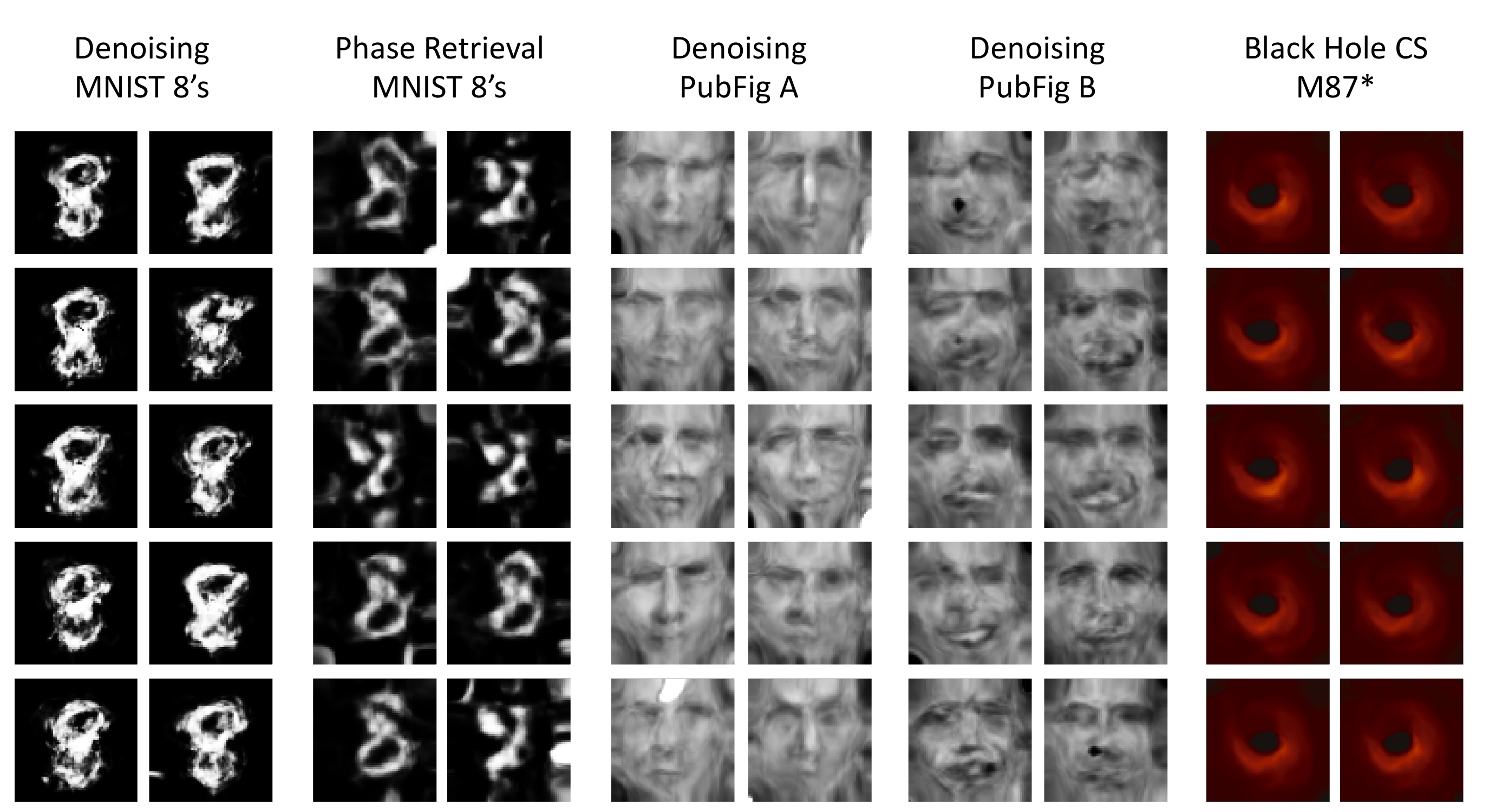}
    \caption{\textbf{Samples from IGM.}  Here we show samples of the IGM defined by $w = G_{\theta}(z)$ where $z \sim \mathcal{N}(0, \sigma^2 I)$. We set $\sigma =1$ for all cases except for phase retrieval, which uses $\sigma=5$. We show samples from a variety of IGMs using different datasets and inverse problems. From left to right: denoising MNIST 8's, phase retrieval measurements of MNIST 8's, denoising celebrity face A, denoising celebrity face B, video reconstruction of a black hole from the black hole compressed sensing problem. The range of these images are all from 0 to 1.}
    \label{fig:generation}
\end{figure*}





\bibliographystylePhys{plain}\bibliographyPhys{supp.bib}

\end{document}